\documentclass[journal]{IEEEtran}

\newif\ifproofs

\proofstrue

\usepackage[numbers,sort&compress]{natbib}
\usepackage{amsmath,amsfonts,amscd,amssymb,dsfont,amsthm} 
\usepackage{cleveref}
\usepackage{bm}
\usepackage{textcomp}
\usepackage{subfig}

\usepackage{enumitem}

\usepackage{graphics} 

\usepackage{algorithm,algorithmicx}
\makeatletter
\renewcommand{\ALG@name}{Algo.}
\makeatother
\usepackage{algpseudocode}

\usepackage{amsthm}

\usepackage{pgfplots}
\pgfplotsset{compat=1.11}
\usepgfplotslibrary{fillbetween}
\usetikzlibrary{intersections}
\pgfdeclarelayer{bg}
\pgfsetlayers{bg,main}

\theoremstyle{definition}

\Crefname{corollary}{Cor.}{Cors.}
\Crefname{equation}{Eq.}{Eqs.}
\Crefname{figure}{Fig.}{Figs.}
\Crefname{tabular}{Tab.}{Tabs.}
\Crefname{table}{Tab.}{Tabs.}
\Crefname{theorem}{Thm.}{Thms.}
\Crefname{definition}{Def.}{Defs.}
\Crefname{section}{Sec.}{Secs.}
\Crefname{proposition}{Prop.}{Props.}
\Crefname{assumption}{Asm.}{Asms.}
\Crefname{example}{Ex.}{Exs.}
\usepackage[colorinlistoftodos,bordercolor=orange,backgroundcolor=orange!20,linecolor=orange,textsize=scriptsize]{todonotes}
\usepackage{multirow}
\usepackage{pbox}
\usepackage{amsfonts}

\newcommand{\norm}[1]{\left\Vert #1\right\Vert}

\newtheorem*{mthm*}{Main Theorem}

\Crefname{mthm}{Main Theorem}{Main Theorems}

\newtheorem{theorem}{Theorem}
\newtheorem{corollary}{Corollary}
\newtheorem{proposition}{Proposition}
\newtheorem{lemma}[theorem]{Lemma}

\newtheorem{assumption}{Assumption}
\newtheorem{example}{Example}

\newtheorem{definition}{Definition}

\Crefname{corollary}{Cor.}{Cors.}
\Crefname{equation}{Eq.}{Eqs.}
\Crefname{figure}{Fig.}{Figs.}
\Crefname{tabular}{Tab.}{Tabs.}
\Crefname{table}{Tab.}{Tabs.}
\Crefname{theorem}{Thm.}{Thms.}
\Crefname{definition}{Def.}{Defs.}
\Crefname{section}{Sec.}{Secs.}
\Crefname{proposition}{Prop.}{Props.}
\Crefname{assumption}{Asm.}{Asms.}
\Crefname{example}{Ex.}{Exs.}

\Crefname{appsec}{Appendix}{Appendices}

\usepackage[colorinlistoftodos,bordercolor=orange,backgroundcolor=orange!20,linecolor=orange,textsize=scriptsize]{todonotes}
\usepackage{multirow}
\usepackage{pbox}
\usepackage{amsfonts}

\renewcommand{\norm}[1]{\left\Vert #1\right\Vert}
\newcommand{\tnorm}[1]{\textstyle\left\Vert #1\right\Vert}

\newcommand{\txt}{\textstyle}

\newcommand{\x}{x}
\newcommand{\xx}{\bm{x}}
\newcommand{\y}{y}
\newcommand{\yy}{\bm{y}}

\newcommand{\zz}{\bm{z}}
\newcommand{\rr}{\bm{r}}

\newcommand{\hf}{\hat{f}}
\newcommand{\g}{\mathbf{g}}
\newcommand{\bh}{\mathbf{h}}

\newcommand{\xag}{X}
\newcommand{\xxag}{\bm{\xag}}
\newcommand{\yag}{Y}
\newcommand{\yyag}{\bm{\yag}}
\newcommand{\zag}{Z}
\newcommand{\zzag}{\bm{\zag}}

\newcommand{\rit}{\mathbb{R}}
\newcommand{\nit}{\mathbb{N}}

\newcommand{\X}{\mathcal{X}}
\newcommand{\FX}{\widetilde{\X}} 
\newcommand{\Sxag}{\overline{\X}} 

\newcommand{\T}{\mathcal{T}}

\newcommand{\I}{\mathcal{I}}
\newcommand{\G}{\mathcal{G}}
\newcommand{\GA}{\mathcal{G}(A)}

\newcommand{\M}{\mathcal{M}}
\newcommand{\N}{\mathcal{N}}

\newcommand{\ww}{\bm{w}}

\newcommand{\cc}{\bm{c}} 
\newcommand{\bv}{\bm{v}}

\newcommand{\Bcuti}{B_\mathbf{u_i}}
\newcommand{\Bcutj}{B_\mathbf{u_j}}

\newcommand{\Bdf}{B_\mathbf{f}}
\newcommand{\Bc}{B_{\cc}} 
\renewcommand{\th}{\theta} 
\renewcommand{\t}{t}
\renewcommand{\i}{i}
\newcommand{\iti}{_{\i,\t}}

\newcommand{\hxx}{\hat{\xx}}

\newcommand{\hxxag}{\hat{\xxag}}

\newcommand{\sxx}{\xx^*}  
\newcommand{\sxxag}{\xxag^*}

\newcommand{\eqd}{\triangleq}
\newcommand{\dth}{\,\mathrm{d}\th}

\newcommand{\snu}{\N} 
\newcommand{\esnu}{^{\snu}}

\newcommand{\dset}{\delta} 
\newcommand{\mdset}{\overline{\delta}} 

\newcommand{\duti}{\lambda} 
\newcommand{\mduti}{\overline{\lambda}} 

\newcommand{\stgccvut}{{\alpha}} 

\newcommand{\rlt}{\text{ri}\,}
\newcommand{\rbd}{\text{rbd}\,}

\newcommand{\aff}{\text{aff}\,}

\newcommand{\bma}{\bm{a}}
\newcommand{\bpsi}{\overline{\psi}}
\newcommand{\tG}{\tilde{\G}}

\newcommand{\hh}{\hspace{-2pt}}
\newcommand{\+}{\hspace{-2pt}+\hspace{-2pt}}

\title{\LARGE \bf
Efficient Estimation of Equilibria of \\ Large Congestion Games with Heterogeneous Players
}

\author{Cheng Wan, Paulin Jacquot, Olivier Beaude, Nadia Oudjane \thanks{Cheng Wan is with Universit\'e Paris-Sud and Inria, Paris, France (\texttt{cheng.wan.2005@polytechnique.org}). Paulin Jacquot is with EDF R\&D (OSIRIS), Inria and \'Ecole polytechnique, CNRS, Palaiseau, France (\texttt{paulin.jacquot@polytechnique.edu}). Olivier Beaude and Nadia Oudjane are with EDF R\&D (OSIRIS), Palaiseau, France (\texttt{\{olivier.beaude,  nadia.oudjane\}@edf.fr}).}
\thanks{
This work was partially supported by the PGMO-ICODE  project ``Jeux de pilotage de flexibilit\'es de consommation \'electrique : dynamique et aspect composite". }
}

\begin{document}

\maketitle

\begin{abstract}
Computing an equilibrium in congestion games can be challenging when the number of players is large. Yet, it is a problem to be addressed in practice, for instance to forecast the state of the system and be able to control it. In this work, we analyze the case of  generalized atomic congestion games, with coupling constraints, and with players that are heterogeneous through their action sets and their utility functions. We obtain an approximation of the variational Nash equilibria---a notion generalizing Nash equilibria in the presence of coupling constraints---of a large atomic congestion game by an equilibrium of an auxiliary population game, where 
 each population corresponds to a group of atomic players of the initial game.
  Because the variational inequalities characterizing the equilibrium of the auxiliary game have smaller dimension than the original problem, this approach enables the fast computation of an estimation of equilibria in a  large congestion game with thousands of heterogeneous players.
\end{abstract}




\section{Introduction}
\paragraph{Motivation}
Congestion games form a class of noncooperative games \cite{nisan2007algorithmic}. In a congestion game, each player chooses a certain quantity of each of the available resources, and pays a cost for each resource obtained by the per-unit cost of that resource multiplied by the quantity she has chosen. 
A congestion game is said to be atomic if there is a finite number of players, and nonatomic if there is a continuum of infinitesimal players. 
The particularity of congestion games is that the per-unit cost of each resource depends only on its total demand. 

Congestion games find practical applications  in various fields such as traffic management \cite{ziegelmeyer2008road}, communications \cite{scutari2012monotone,altman2006survey} and more recently in electrical systems \cite{mohsenian2010autonomous,PaulinTSG17}.

The concept of Nash equilibrium (NE) \cite{nash1950equilibrium} has emerged as the most credible outcome in the theory of noncooperative games. However, it is shown that computing a NE, when it exists, is a hard problem  \cite{ackermann2008impact,fabrikant2004complexity}.
 NEs are often characterized by some variational inequalities. Therefore, the efficiency of the computation of NEs  depends on the dimension of the variational inequalities in question, hence on the number of players and the number of constraints.   
The problem can be intractable at a large scale, when considering several thousands of heterogeneous agents, which  is often the case when describing real situations. The case of generalized Nash equilibria \cite{harker1991generalized}, when one considers coupling constraints---for instance capacity constraints---makes the problem even harder to solve. Meanwhile, coupling constraints commonly exist in real world. For example, in transportation, roads and communication channels have a limited capacity that should be considered. In the energy domain, production plants are also limited in the magnitude of variations of power, inducing some ``ramp constraints'' \cite{carrion2006computationally}.

However, estimating the outcome situation---supposed to correspond to an equilibrium---is often a priority for the  operator of the system. For instance, the operator controls some variables such as  physical or managerial parameters,  of a communication or transport network and wishes to optimize the performance of the system. 
 The computation of equilibria or their approximation is also a key aspect in bi-level programming \cite{ColsonMS2007bilevel}, where the lower level corresponds to a usually large scale game, and the upper level corresponds to a decision problem of an operator choosing optimal parameters. These parameters, such as prices or taxes, are to be applied in the low level game, with the aim of maximizing the revenue in various industrial sectors and public economics, such as highway management, urban traffic control, air industry, freight transport and radio network \cite{LabbeMS1998, BrotcorneLMS2000, BrotcorneLMS2001, CoteMS2003, elias2013joint}.

In this paper, we consider atomic congestion games with a finite but large number of players. We propose a method to compute an approximation of NEs, or variational Nash equilibria (VNEs) \cite{harker1991gne} in the presence of coupling constraints. 
The main idea is to reduce the dimension of the variational inequalities characterizing NEs or VNEs. 
The players are divided into groups with similar characteristics. Then, each group is replaced by a homogeneous population of nonatomic players. 
To provide an estimation of the equilibria of the original game, we compute a Wardrop equilibrium (WE) \cite{wardrop1952some} in the approximating nonatomic population game, or a variational Wardrop equilibrium (VWE) in the case of coupling constraints. 
The quality of the estimation depends on how well the characteristics, such as action set and cost function, of each homogeneous population approximate those of the atomic players it replaces.  

In addition to the reduction in dimension, another advantage of WE is that it is usually unique in congestion games, in contrast to NEs. 
In bi-level programming, the uniqueness of a low level equilibrium allows for clear-cut comparative statics and sensitivity analysis at the high level.

\paragraph{Related works}

The relation between NEs in large games and WEs has been studied in  Gentile et al.~\cite{gentile2017nash}. In their paper, the authors also consider atomic congestion games with coupling constraints and show, using  variational inequalities approach, that the distance between a NE and a WE converges to zero when the number of players tends to infinity. 
Their WE corresponds to an equilibrium of the game where each atomic player is replaced by a population.
The objective of our paper is different. 
We look for an approximation of NEs by reducing the dimension of the original game. To this end, we regroup many players into few homogeneous populations. 
 Our results apply to the  subdifferentiable case in contrast to the differential case considered in \cite{gentile2017nash}.

In \cite{PaulinWan2018}, Jacquot and Wan show that, in congestion games with a continuum of heterogeneous players, the WE can be approximated by a NE of an approximating game with a finite number of players. 
In \cite{PaulinWan2018nonsmooth}, those results are extended to aggregative games, a more general class of games including congestion games, furthermore with nonsmooth cost functions.

Different algorithms have been proposed to solve monotone variational inequalities corresponding to NE or WE, such as  \cite{cohen1988auxiliary,fukushima1986relaxed,zhu1993modified,facchinei2007finite,facchinei2010gnep}, and more recently \cite{yi2018asynchronous,yi2017distributed,parise2017distributed,tatarenko2018learning} and the references therein.

 The approach developed in the present paper is actually the inverse of the one taken in \cite{PaulinWan2018} and \cite{PaulinWan2018nonsmooth}: here, the WE in the auxiliary game serves as an approximation of an NE of the original large game.

\paragraph{Main contributions}
The contributions of this paper are the following.
\begin{itemize}[wide]
\item We define an approximating population game (\Cref{subsec:class}). The idea is that the auxiliary game has smaller dimension but is close enough to the original large game---quantified through the Hausdorff distance between action sets and between subgradients of players' objective functions. 
\item We show theoretically that a particular variational Wardrop equilibrium (VWE) of the  approximating population game is close to any variational Nash equilibria (VNE) of the original game with or without coupling constraints, while the computation of the former is much faster than the later because of the dimension reduction. We provide an explicit expression of the error bound of the approximating VWE (\Cref{th:main}).
\item We give auxiliary results on variational equilibria:  when the number of players is large, VNEs are close to each other (\Cref{th:VNEareclose}) and that VNEs are close to the approximating VWE (\Cref{th:bound}). This last theorem extends \cite[Thm.~1]{gentile2017nash} in the case of nondifferentiable cost functions, in the framework of congestion games.
\item Last, we provide a numerical illustration of our results (\Cref{sec:appli}) based on a practical application: the decentralized charging of electric vehicles through a demand response mechanism \cite{palensky2011demand}. 
This example illustrates the nondifferentiable case through piece-wise linear electricity prices (``block rates tariffs''), with coupling constraints of capacities and limited variations on the aggregate load profile  between time periods. 
This example shows that the proposed method is implementable and that it reduces the time needed to compute an equilibrium by computing its approximation (six times faster for an approximation with a relative error of less than $2\%$).
\end{itemize}

The remainder of this paper is organized as follows:  \Cref{sec:congestionModel} specifies the framework of congestion games with coupling constraints, and recalls the notions of variational  equilibria and monotonicity for variational inequalities, as well as several results on the existence and uniqueness of equilibria. 
\Cref{sec:approxRes} formulates the main results: \Cref{subsec:class} shows that a VWE approximates VNEs in large games and then, \Cref{subsec:class} formulates the approximating population game with the approximation measures, and gives an error bound on the VWE of the approximating game with respect to the original VNEs. 
\Cref{sec:appli} presents a numerical illustration in the framework of demand response for electric vehicle smart charging.

\section{Congestion Games with Coupling Constraints}  \label{sec:congestionModel}

\subsection{Model and equilibria}
The original game throughout this paper is an atomic splittable  congestion game, a particular sort of aggregative games where a set of resources is shared among finitely many players, and each resource incurs a cost increasing with the aggregate demand for it. The formal definition is as follows.
 \begin{definition}\label{def:atomicGame} 
An \emph{atomic splittable congestion} game $\G$ is defined by: 
\begin{itemize}[leftmargin=*,wide,labelindent=-1pt]
\item a finite set of players: $\I=\{1,\dots, i, \dots ,I\}$,
\item a finite set of resources: $\T=\{1,\dots, t, \dots,T\}$,
\item for each resource $\t$, a cost function $c_t: \rit_{+} \rightarrow \rit$,
\item for each player $\i$, a set of feasible choices: $\X_i \subset \rit^T_+$, an element $\xx_\i=(x\iti)_{t\in \T} \in \X_i$  signifies that $i$ has demand $\x_{i,t}$ for resource $t$,
\item for each player $\i$, an individual utility function $u_i: \X_i \rightarrow \rit$,
\item a coupling constraint set $A\subset \rit^T$.
\end{itemize}

We denote  by $\FX\eqd \X_1\times\dots \times \X_I$ the product set of action profiles. An action profile $\xx= (\xx_\i)_{\i\in\I}\in \FX$ induces a profile of aggregate demand for the resources, denoted by $\xxag=(\xag_t)_{t\in \T}\eqd (\sum_{\i\in\I}x\iti)_{t\in \T}$. 
We denote the set of feasible aggregate demand profiles  by:\begin{equation*}\Sxag\eqd \{ \xxag \in \rit^T  : \forall i\in \I,  \exists \xx_i \in \X_i\, \text{ s.t. }\txt\sum_{\i\in\I} \xx_\i= \xxag  \} \ . 
\end{equation*} 
With coupling constraints, the set of feasible aggregate demand profiles with coupling constraints is $\Sxag\cap A$, and the set of feasible action profiles is denoted by $\FX(A)=\{\xx \in \FX: \txt\sum_{\i\in\I} \xx_\i \in \Sxag\cap A\}$. 

Let the vector of cost functions denoted by $\cc(\xxag)=(c_t(\xag_t))_{t\in \T}$, where $c_t(\xag_t)$ is the (per-unit of demand) cost of resource $t$  when the aggregate demand for it is $\xag_t$.

Player $i$'s cost function $f_\i: \X_i \times \Sxag \rightarrow \rit$ is defined by:
\begin{equation}\label{eq:cost_player_def}
f_\i(\xx_\i,\yyag)=\sum_{\t\in\T}\x\iti c_t(Y_t) -u_\i(\xx_\i),\; \forall \xx_\i \in \X_i, \, \yyag\in \Sxag.
\end{equation} 
Given $\xx_{-i}\in \prod_{j\neq i}\FX_j$ and $\xxag_{-\i}\eqd \sum_{j\neq i}\xx_j$, player $i$'s cost is $f_\i(\xx_\i,\xx_i+\xxag_{-\i})$,
composed of the network costs and her individual utility. 

This atomic congestion game with coupling constraints is defined as the tuple $( \I, \T, \FX , A, \cc, (u_i)_{i\in\I}) $.
\end{definition}
\smallskip
%


In an atomic game, there are finitely many players whose actions are not negligible on the aggregate profile and on the objectives of other players. The term ``atomic'' is opposed to ``nonatomic'' where players have an infinitesimal weight \cite{nisan2007algorithmic}. 
The term ``splittable'' refers to the infinite number of choices of pure actions $\xx_i\in \X_i$ for each player $i$, as opposed to the unsplittable case where each player can only choose one action in a finite subset of $2^\T$ \cite{Rosenthal1973}. 
Besides, atomic splittable congestion games are particular cases of aggregative games \cite{gentile2017nash}: each player's cost function depends on the actions of the others only  through the aggregate profile $\xxag$.  
\smallskip

The following standard assumptions are adopted in this paper. 

\begin{assumption}\label{assp_convex_costs}~~
\nopagebreak

(1) For each player $\i\in\I$, the set $\X_\i$ is a convex and compact subset of $\rit^T$ with nonempty relative interior. 

(2) The cost function $c_t$ for each resource $t\in \T$ is continuous, convex and  non-decreasing on $(-\eta,+\infty)$ for a positive $\eta > 0$. 

(3) For each player $\i\in\I$,  individual utility function $u_\i$ is continuous and concave in $\xx_\i$ on $\X_i$.

(4) $A$ is a convex closed set of $\rit^T$, and $\Sxag\cap A$ is not empty.
\end{assumption}

%


An important class of atomic splittable congestion games corresponds to the case where resources constitute a \emph{parallel-arc} network \cite{orda1993competitive}. There, each player $i$ has a total demand and specific bounds on the demand that she can have for each resource so that her strategy set is given by 
$\X_\i=\{\ \xx_\i \in \rit^T_+  :  \txt\sum_\t  \x\iti =m_i \text{ and }  \underline{x}\iti \leq \x\iti \leq \overline{x}\iti\}$. Here, $m_i$ can represent the mass of data to send over different canals, or the amount of energy to consume over several time periods \cite{PaulinTSG17}. In particular, the demand is continuous and splittable, as the mass $m_i$ is split over the resources $t\in\T$.
 
 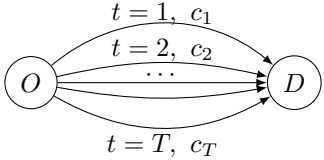
\begin{figure}[ht!]
\begin{center}
\begin{tikzpicture}[scale=0.35]
\node[draw,circle,scale=1] (a1)at(-5,0) {$O$};
\node[draw,circle,scale=1] (b1)at(5,0) {$D$};
\draw[->,>=latex] (a1) to[bend left=40] (b1);
\node[scale=1] (1) at (0,2.6) {$t=1, \ c_1$};
\node[scale=1] (1) at (0,1.25) {$t=2, \ c_2$};
\draw[->,>=latex] (a1) to[bend left=13] (b1);
\node[scale=1] (1) at (0,0.3) {$\cdots$};
\draw[->,>=latex] (a1) to[bend right=30] (b1) ;
\draw[->,>=latex] (a1) to[bend right=10] (b1);
\draw[->,>=latex] (a1) to (b1);
\node[scale=1] (1) at (0,-2.4) {$t=T, \ c_T$};
\end{tikzpicture}
\caption{A parallel-arc network with $T$ arcs/resources.\label{fig:network}}
\end{center}
\end{figure}

\vspace{-0.5cm}

Our model is more general than the one in \cite{orda1993competitive}, not only because the network topology can be arbitrary, but also because it allows for \emph{elastic demands} from the players.  
For example, player $i$'s action set can be 
$\X_\i=\{\ \xx_\i \in \rit^T_+  :  \underline{m}_i \leq \txt\sum_\t  \x\iti \leq \overline{m}_\i \}$. Indeed, the individual utility function counterbalances the network cost: a player may be willing to pay more congestion cost by increasing the demand, because she profits from a higher individual utility, and vice versa.

Let us cite two common forms of \emph{individual utility function}. The first one measures the distance between a player's choice and her preference $\yy_\i\in \X_i$: $u_i(\xx_\i)= -\omega_i \tnorm{\xx_\i-\yy_\i}^2 $, where $\omega_i>0$ is the value that the player attaches to her preference. 
The second one is
$u_i(\xx_\i)= \omega_i\log\left(1+\txt\sum_{t} \x\iti \right)$, which is increasing in the player's total demand.

Finally, in congestion games, \emph{aggregate constraints} are very common. For example, in routing games, there can be a capacity constraint linked to each arc. In energy consumption games, due to the operational constraints of the power grid, there can be both minimum and maximum consumption level for each time slot, and ramp constraints on the variation of energy consumption between time slots. This is why congestion games with aggregate constraints are of particular interest.
%
\medskip


To separate player $i$'s choice from those of the other players in her cost function, define 
\begin{equation*}
\hat{f}_i(\xx_i,\yyag) \eqd f_i(\xx_i, \yyag+\xx_i)
\end{equation*}
for $\xx_i$ in $\X_i$ and $\yyag$ in $\Sxag_{-i}=\{\sum_{j\in\I\setminus\{i\}}\xx_j:\xx_j\in \X_j\}$.


Since $\cc$ and $u_i$'s are not necessarily differentiable, we need to define the subdifferential of the players' utilities w.r.t. their actions for the characterization of equilibrium. 

%
%
%
%
Let us define two correspondences, $H$ and $H'$, from $\FX$ to $\rit^{IT}$: for any $\xx\in \FX$,
\begin{align*}
H(\xx) &\eqd\{(\bm{h}_i)_{i\in \I}\in \rit^{IT}: \bm{h}_i \in \partial_1 \hf_\i(\xx_\i, \xxag_{-\i}) , \  \forall i\in \I\} \\
&= \txt\prod_{i\in\I} \partial_1 \hf_\i(\xx_\i, \xxag_{-\i}) \  ; \\
H'(\xx) & \eqd\{(\bm{h}_i)_{i\in \I}\in \rit^{IT}: \bm{h}_i \in \partial_1 f_\i(\xx_\i, \xxag) , \  \forall i\in \I\} 
\\ &= \txt\prod_{i\in\I} \partial_1 f_\i(\xx_\i, \xxag)\  ,
\end{align*}
where $\partial_1$ signifies the partial differential w.r.t. the first variable of the function.
%
The interpretation of $H(\xx)$ is clear: $\bm{h}_i$ is a subgradient of player $i$'s utility function $\hf_i$ w.r.t. her action $\xx_i$. Let us leave the interpretation of $H'(\xx)$ till \Cref{def:pseudoVNE}. 
For the moment, let us write the explicit expression of $H$ and $H'$:
\begin{lemma}\label{prop:subgradients-sets}
For each $\xx \in \FX$:\\
\textbullet~ $\bm{h} \in H(\xx)$ if and only if there are $\g'_i\in \partial (-u_\i)(\xx_i)$ and $\bma_i \in  \prod_{t\in \T} \partial c_t(\xag_t)$  $\forall i$ s.t.:
\begin{equation*}
\bm{h}_i= \cc(\xxag)+ (\x_{i,t} a_{i,t})_t +  \g'_i \ ,\quad \forall i\in \I\  ;
\end{equation*} 
\textbullet~ $\bm{h}'\in H'(\xx)$ if and only if there is $\g'_i\in \partial (-u_\i)(\xx_i)$ $\forall i$ s.t.:
\begin{equation*}
\bm{h}'_i= \cc(\xxag) +  \g'_i \ ,\quad \forall i\in \I\ .
\end{equation*}
where  $\partial (-u_\i)(\xx_i)$ is the subdifferential of convex function $-u_i$ at $\xx_i$ and $\partial c_t(\xag_t)$ the subdifferential of $c_t$ at $\xag_t$.
\end{lemma}
\begin{proof}
See \Cref{app:propSubgradients}.
\end{proof}
\smallskip

In our framework with coupling constraints, the notion of Nash equilibrium (NE) \cite{nash1950equilibrium} is replaced by that of Generalized Nash  Equilibrium (GNE): $\xx\in \FX(A)$ is a GNE if, for each player $i$, $\hf_i(\xx_i, \xxag_{-i})\leq \hf_i(\yy_i, \xxag_{-i})$ for all $\yy_i$ s.t. $\yy_i+\xxag_{-i}\in \Sxag\cap A$. For atomic games, a special class of GNE is called Variational Nash Equilibria \cite{harker1991generalized,Kulkarni2012vne}, which enjoys some symmetric properties and can be easily characterized as the solution of the VI \eqref{cond:ind_opt_ve} below.

\begin{definition}[Variational Nash Equilibrium (VNE), \cite{harker1991gne}]\label{def:ve-finite}
A VNE   is a solution $\hxx\in \FX(A)$ to the following GVI problem: 
 \begin{align}\label{cond:ind_opt_ve}
\exists\, \g\in H(\hxx)  \text{ s.t. } 
&  \txt \big\langle \g, \xx- \hxx\big\rangle\geq 0,\; \forall \xx \in \FX(A).
 \end{align}
In particular, if  $\Sxag\subset A$,  a VNE is a NE.
\end{definition}
In this paper, we adopt VNE  as the equilibrium notion in the presence of aggregate constraints. 

As the  first step of approximation, let us define a nonatomic congestion game $\GA'$ associated to $\GA$. Let  each player $i$ be replaced by a continuum of identical nonatomic players, represented by  interval $[0,1]$ with each point thereon corresponding to a nonatomic player. Each player in population $i$ has action set $\X_i$ and individual utility function $u_i$. 
\begin{definition}\label{def:pseudoVNE}
A \emph{symmetrical variational Wardrop equilibrium} (SVWE) of $\GA'$  is  a solution to the following GVI:
\begin{align}\label{eq:def-pseudo}
  \exists \g\in H'(\sxx) \text{ s.t. } 
 & \langle \g, \xx-\sxx \rangle\geq 0,\; \forall \xx\in \FX(A)\ .
 \end{align}
\end{definition}
For the definition of variational Wardrop equilibrium (VWE) and further discussion, we refer to \cite{PaulinWan2018nonsmooth}. In particular, a VWE is characterized by an infinite dimensional variational inequality. Here, we consider only those VWE where all the nonatomic players in population $i$ take the same action $\xx_i$. Such a SVWE exists because the players are identical in the same population. 

The second interpretation of SVWE is the following: when the number of players is very large so that the individual contribution of each player on the aggregate action $\xxag$ is almost negligible, the term $x_{i,t}a_{i,t}$ in $\bm{h}\in H(\xx)$ is so small that $H(\xx)$ can be approximated by $H'(\xx)$ (cf. \Cref{prop:subgradients-sets}). This is the interpretation adopted in  \cite{gentile2017nash}. 
However, note that a SVWE of $\GA'$ is not an equilibrium of $\GA$ in the sense of a ``stable state'' for the atomic congestion game. 

\smallskip

The existence of equilibria defined in \Cref{def:ve-finite,def:pseudoVNE} are obtained without more conditions than \Cref{assp_convex_costs}:
\begin{proposition}[Existence of equilibria]\label{prop:exist_ve}
  Under \Cref{assp_convex_costs}, $\GA$ (resp. $\GA'$) admits a VNE (resp. SVWE). 
\end{proposition}
\proof: see \Cref{app:proof:exist_ve}.
\smallskip

Before discussing the uniqueness of equilibria, let us recall some relevant monotonicity assumptions.
\begin{definition}\label{def:mono_non}
A correspondence $\Gamma:\FX \rightrightarrows \rit^T$ is:
\begin{itemize}[wide,leftmargin=0pt,labelindent=0pt]
\item\emph{monotone} if for all $\xx, \yy \in \FX , \g\in \Gamma(\xx), \bh \in \Gamma(\yy)$:
\begin{equation}\label{cd:mono_non}
\txt\sum_{i\in \I} \langle \g_i - \bh_i, \xx_i - \yy_i \rangle \geq 0\ ;
\end{equation}

\item \emph{strictly monotone} if the equality in \eqref{cd:mono_non} holds \textit{iff} $\xx=\yy$;

\item \emph{aggregatively strictly monotone} if the equality in \eqref{cd:mono_non} holds \textit{iff} $\sum_i\xx_i=\sum_i\yy_i$;
\item $\stgccvut$-\emph{strongly monotone} if $\stgccvut>0$ and, for all $\xx, \yy \in \FX$:
\begin{equation}\label{cd:strong_mono_non}
\sum_{i\in \I} \langle  \g_i -\bh_i, \xx_i - \yy_i \rangle \hh\geq\hh \stgccvut\|\xx-\yy\|^2, \, \forall\g\!\in\! \Gamma(\xx), \bh\! \in\! \Gamma(\yy)\, ;
\end{equation}

\item $\beta$-\emph{aggregatively strongly monotone} on $\FX$ if $\beta>0$ and, for all $\xx, \yy \in \FX$ with $\xxag=\sum_i\xx_i$, $\yyag=\sum_i\yy_i$:
\begin{equation}\label{cd:strong_agg_mono_non}
 \sum_{i\in \I} \langle \g_i - \bh_i, \xx_i - \yy_i \rangle\! \geq\! \beta\|\xxag -\yyag\|^2 , \, \forall\g\!\in\! \Gamma(\xx), \bh\! \in\! \Gamma(\yy)\, .
\end{equation}
\end{itemize}

\end{definition}

If $T=1$,  ``monotone" corresponds to ``increasing". 
Besides, (aggregatively) strict monotonicity implies monotonicity, while strong (resp. aggregatively strong) monotonicity implies strict (resp. aggregatively strict) monotonicity.
\smallskip

%


In \Cref{th:unique_vwe} below, we recall some existing results concerning the uniqueness of VNE and SVWE, according to the  monotonicity of  $H$ and $H'$:
\begin{proposition}[Uniqueness of equilibria]\label{th:unique_vwe}%
Under \Cref{assp_convex_costs}: \\
(1) if $H$ (resp. $H'$) is strictly monotone, then $\GA$ (resp. $\GA'$) has a unique VNE (resp. SVWE); \\
(2) if  $H$ (resp. $H'$) is aggregatively strictly monotone, then all VNE (resp. SVWE) of $\GA$ (resp. $\GA'$) have the same aggregate profile;\\
(3) if $H$ (resp. $H'$) is only aggregatively strictly monotone but, in addition, for each $\i\in \I$, $u_\i(\xx)$ is strictly concave, then there is at most one NE (resp. WE) in the case without aggregative constraint.
\end{proposition}
\proof see \Cref{app:proof-uniquenessVWE}.
\smallskip

\Cref{prop:monotonemap} below gives sufficient conditions for the (strong) monotonicity to hold for $H'$.
\begin{proposition}[Monotonicity of $H'$]\label{prop:monotonemap}
Under \Cref{assp_convex_costs}, \\
(1)  $H'$  is monotone.\\
(2) If for each $i \in \I$, $u_i$ is $\stgccvut_\i$-strongly concave,  then $H'$ is $\stgccvut$-strongly monotone with $\stgccvut\eqd \min_{\i\in\I}\stgccvut_\i$.\\
(3) If for each $t\in\T$,  $c_t$ is $\beta_t$-strictly increasing, then $H'$ is $\beta$-aggregatively strongly monotone with  $\beta\eqd \min_{t\in\T}\beta_t$.
\end{proposition}
\proof See \Cref{app:proof-monotonicity}.
\medskip

As opposed to the monotonicity of $H'$ shown in \Cref{prop:monotonemap}, $H$ is rarely monotone
(except in some particular cases, e.g. with $\cc$ linear \cite{orda1993competitive,richman2007topounique}): even in the case where $\cc$ is piece-wise linear, the \Cref{ex:nonlinearH} below shows that $H$ can be non monotone.
\begin{example} \label{ex:nonlinearH}
Let $I=2$ and $T=1$, $\X_1=\X_2=[0,4]$. Consider the cost function $c(X)=X$ for $X \leq 4$ and $c(X)=3X-8$ for $X\geq 4$. \Cref{assp_convex_costs} holds. Consider the profiles $\x_1=3, \x_2=1$ and $\y_1=4,\y_2=0$, then $\g\eqd (c(4)+3\x_1,c(4)+3\x_2) \in H(\xx)$ and $ \bm{h} \eqd (c(4)+1 \y_1,c(4)+1 \y_2)\in H(\yy)$, but:
\begin{equation*}
\txt\sum_{i\in\{1,2\} } \langle \g_{i}-\bm{}h_{i} , \x_i-\y_i \rangle = -2 <0 \ .
\end{equation*}
\end{example}

In view of \Cref{th:unique_vwe}, the absence of monotonicity of $H$ can result in multiple VNEs \cite{bhaskar2009notunique}.

In \cite{PaulinTSG17}, a particular case with parallel arc network is shown to have a unique NE. However, in the next section, we shall prove that, when the number of players is very large, all VNEs are close to each other and they can be well approximated by the unique SVWE.

\section{Approximating VNEs of a large game}
\label{sec:approxRes}
\subsection{Considering SVWE instead of VNE}

The approximation of VNEs is done in two steps.  The first step consists in replacing VNE by SVWE. According to the second interpretation of SVWE, the SVWE should be close to the VNEs in a large game. Now, let us formulate this idea and bound the distance between the two.
\smallskip

Denote by $\X_0\subset \rit^T$ the convex closed hull of $\bigcup_{i\in \I}\X_i$, by $m=\max_{\xx \in \X_0}\|\xx\|$ 
and $M=I m\geq \max_{\xxag\in \Sxag}\|\xxag\|$. 
Define compact set $\M=[0,M+\delta I]^T$, where $\delta>0$ is a constant to be specified later.

Denote by $C\eqd \sup\{b\in \rit: b\in \partial c_t(\xag_t), \xxag\in \Sxag, t\in \T\}$ the upper bound on the subgradients of $\cc$. 

\smallskip

\Cref{th:unique_vwe} and \Cref{prop:monotonemap} show that, in general, VNEs are not unique. However, when the set of players is large, 
 VNEs are indeed close to each other:
\begin{theorem}[VNEs are close to each other]\label{th:VNEareclose}
Under \Cref{assp_convex_costs}, let $\xx$ and $\yy$ in $\FX(A)$ be two distinctive VNEs of $\G(A)$. Then \\
(1) if for each $i\in \I$, $u_i$ is $\stgccvut_i$-strongly concave, then: 
\begin{equation}
  \|\xx-\yy\| \leq 2M \sqrt{\txt\frac{TC}{\stgccvut I}} \ ,
  \end{equation} with $\stgccvut \eqd \min_{i\in \I} \stgccvut_i$;\\
(2) if for each $t\in \T$, $c_t$ is $\beta_t$-strictly increasing, then: 
\begin{equation}
  \|\xxag-\yyag\|\leq  2M \sqrt{\txt\frac{TC}{\beta I}} \ ,
  \end{equation} with $\beta\eqd \min_{t\in \T} \beta_t$.
\end{theorem}
\proof: See \Cref{app:proof-VNEclose}.
\smallskip

The first step of approximation is based upon the following theorem which gives an upper bound on the distance between a VNE and the unique SVWE.
\begin{theorem}[SVWE is close to VNE]\label{th:pseudoVNE}
Under \Cref{assp_convex_costs}, let $\xx\in \FX(A)$ be a VNE of  $\GA$ and $\sxx\in \FX(A)$ a SVWE of $\GA'$, then: \\
(1) if for each $i\in \I$, $u_i$ is a  $\stgccvut_i$-strongly concave, then $\xx^*$ is unique and: \begin{equation}
  \|\xx-\sxx\| \leq M\sqrt{\tfrac{2TC}{\stgccvut I}} \ , 
  \end{equation}
   with $\stgccvut \eqd  \min_{i}\stgccvut_i $ ;\\
(2) if for each $t\in \T$, $c_t$ is $\beta_t$-strictly increasing, then $\xxag^*$ is unique and: 
\begin{equation}
\|\xxag-\sxxag\|\leq  M \sqrt{\tfrac{2TC}{\beta I}} \ ,
\end{equation}
 with $\beta\eqd \min_{t} \beta_t$.
\end{theorem}
\begin{proof} Similar to the proof of \Cref{th:VNEareclose}. \end{proof}

An upper bound on the distance between two VNEs can also be derived from \Cref{th:pseudoVNE}, applying the triangle inequality. However, \Cref{th:VNEareclose} gives a tighter upper bound. 

\Cref{th:pseudoVNE} shows that, if the number of players $I$ is large, then the SVWE will provide a good approximation of a VNE of $\GA$. Similar results are obtained in \cite{gentile2017nash}. However, this does not reduce the dimension of the GVI to resolve: the GVI characterizing the VNE and those characterizing the SVWE have the same dimension. 
For this reason, the second step of approximation consists in  regrouping similar populations.

\subsection{Classification of populations}
\label{subsec:class}

In this subsection, we shall 
 regroup the populations in $\GA'$ with similar strategy sets $\X_i$ and utility subgradients $\partial(-u_i)$ (w.r.t. the Hausdorff distance, denoted by $d_H$) into  larger populations, endow them with a common strategy set and a common utility function, so that the SVWE of this new population game approximates the SVWE of $\GA$.

At the SVWE of the new population game with a reduced dimension, all the nonatomic players in the same population play the same action, by the definition of SVWE. Therefore, in order for this new SVWE to well approximate the SVWE in $\GA'$, we must ensure that populations with similar characteristics in $\GA'$ do play similar actions at the SVWE of $\GA'$. 
\Cref{prop:continuiteNE} formulates this results in the case without coupling constraint. 

Without loss of generality, we assume that for each $i\in \I$, $u_i$ can be extended to a neighborhood of $\frac{\M}{I}=[0,m]^T$, and is bounded on $\frac{\M}{I}$. Denote $\Bcuti=\sup\{\|\g'_i\|: \g'_i \in \partial (-u_i)(\xx_i), \xx_i\in \X_i\}$, and $\Bc=\sup\{\|\cc(\xxag)\|: \xxag\in \Sxag\}$.

\begin{proposition}\label{prop:continuiteNE}
Under \Cref{assp_convex_costs}, let $\sxx\in \FX$ be a SVWE of $\G'$ (without coupling constraints).  For two populations $i$ and $j$ in $\I$, if  $u_i$ is $\stgccvut_i$-strongly concave, $d_H(\X_\i, \X_{j}) \leq \dset$,  and $\sup_{\xx_j \in \X_j } \sup_{\g'_j\in \partial (-u_j)(\xx_j)} d(\g'_j, \partial (-u_i)(\xx_j))\leq \duti$, then 
\begin{equation*}
\norm{\sxx_\i - \sxx_j}^2\leq \tfrac{1}{\stgccvut_i}\big((\Bcuti+\Bcutj+2\Bc)\duti+2\dset m\big).
\end{equation*}
\end{proposition}
\proof: See \Cref{app:proof:lmcontinuite}.
%
%

In the case with coupling constraints, the proof for a similar result is more complicated. Let us leave it to \Cref{cor:ij}. 

Let us now present the regrouping procedure.
\smallskip

Denote an auxiliary game $\tG\esnu(A)$, with a set $\N$ of $N$ populations. Each population $n\in \N$ corresponds to a subset of populations in game $\GA'$, denoted by $\I_n$, and $\dot\bigcup_{n\in \N}\I_n = \I$. 

Denote $I_n = |\I_n|$ the number of original populations now included in $n$. By abuse of notations, let $n$ also denote the interval $[0, I_n]$, so that each nonatomic player in population $n$ is represented by a point $\theta\in [0,I_n]$. 
 The common action set of each nonatomic player in $n$ is a compact convex subset of $\rit^T$, denoted by $\X_n$. 
 
Each player $\th$ in each population $n$ having chosen action $\xx_\th$, let $\xxag \eqd \sum_{n\in \N}\int_{\th\in n} \xx_\th \dth$ denote the aggregate action profile. The aggregate action-profile set in $\rit^T$ is then:\\
$\Sxag\esnu=\{\sum_{n\in \N}\int_{\th\in n} \xx_\th \dth: \xx_\th\in \X_n, \, \forall \th\in n, \, \forall n\in \N\}$. 
The cost function of player $\th$ in population $n\in \N$ is:
\begin{equation*}
f_n(\xx_\th, \xxag)=\langle \xx_\th, \cc(\xxag) \rangle - u_n(\xx_\th), 
\end{equation*}
where the common individual utility function $u_n$ for all the players in $n$ is concave on a neighborhood of $\X_n$.

We are only interested in \emph{symmetric} action profiles, i.e. where all the nonatomic players in the same population $n$ play the same action.  Denote the set of symmetric action profiles by $\FX\esnu=\prod_{n\in \N} \X_n$. Let us point out that a symmetric action profile happens as a specific case in the non-cooperative game, without any coordination between the players within a population. Besides, considering the coupling constraint that $\xxag\in A$, we define $\FX\esnu(A)=\{\xx\in\FX\esnu:\xxag=\sum_{n\in \N}I_n\xx_n\in A \}$.
\smallskip

%
%
%
%

Let us introduce two indicators to ``measure'' the quality of the clustering of $\tG\esnu$:
\begin{itemize}
\item $\mdset = \max_{n\in \N} \dset_n$, where
\begin{equation} \label{eq:def_dset}
\dset_n \eqd \txt\max_{\i \in \I\esnu_n} d_{H}\left( \X_\i, \X_n \right)\ ,
\end{equation}
\item $\mduti = \max_{n\in \N} \duti_n$, where
\begin{equation}\label{eq:def_dut}
\duti_n \eqd  \max_{\i \in \I_n}\sup_{\xx\in \X_n}   d_H\left(  \partial (-u_n)(\txt {\xx }) , \partial (-u_i)(\xx) \right).
\end{equation}
\end{itemize}
The quantity $\dset_n$ measures the heterogeneity in strategy sets of populations within the group $\I_n$, while  $\duti_n$ measures the heterogeneity in the subgradients in the group $\I_n$.

Since the auxiliary game $\tG\esnu$ is to be used to compute an approximation of an equilibrium of the large game $\G$, the indicators $\dset_n$ and $\duti_n$ should be minimized when defining $\tG\esnu$. 
Thus, we assume that $(\X_n)_n$ and $(u_n)_n$ are chosen such that the following holds:
\begin{assumption} \label{assp:Xnandun} For each $n \in\N$, we have:
\begin{enumerate}[wide]
\item $\X_n$ is in the convex hull of $\bigcup_{i\in \I_n}\X_i$, so that $ \norm{\X_n} \leq m$. 
Moreover, for each $i\in\I_n$, $\aff \X_i \subset \aff \X_n$, where $\aff S$ denotes the affine hull of set $S$;
\item similarly,  $u_n$ is  such that $\partial u_n(\xx)$ is contained in the convex hull of $\bigcup_{i\in \I_n}\partial (-u_i)(\xx)$ for all $\xx\in \X_n$, so that $\norm{\partial (-u_n)}_\infty \leq \max_{i\in \I_n}\Bcuti$.
\end{enumerate}
\end{assumption}

An interesting case in the perspective of minimizing the quantities $\mdset$ and $\mduti$ is  when $\I$ can be divided into homogeneous populations, as in \Cref{ex:homogenousPops} below.

\begin{example} \label{ex:homogenousPops}
The player set $\I$ can be divided into a small number $N$  of subsets $(\I_n)_n$, with homogeneous players inside each subset $\I_n$ (i.e., for each $n$ and $i,j \in \I_n$, $\X_i=\X_j $ and $u_i=u_j$). In that case, consider an auxiliary game $\tG\esnu$ with $N$ populations and, for each $n\in\N$ and $i\in\I_n$, $\X_n^N\eqd \X_i$ and $u_n \eqd u_i$. Then, $\mdset=\mduti=0$.
\end{example}


In order to approximate the SVWE of $\GA'$ by the SVWE of an auxiliary game $\tG\esnu$, let us first state the following result on  the geometry of the action sets for technical use.
\begin{lemma}\label{lm:intprofile}
Under \Cref{assp_convex_costs}, there exists a strictly positive constant $\rho$  and an action profile $\zz\in \FX(A)$ such that, $d(\zz_\i, \rbd \X_\i)\geq \rho$ for all $\i\in \I$, where $\rbd$ stands for the relative boundary.
\end{lemma}
\begin{proof} See \Cref{app:proof:lm-intprofile}. \end{proof}
\Cref{lm:intprofile} ensures the existence of a profile $\zz$ such that $\zz_i$ has uniform distance to the relative boundary of $\X_i$ for all $i$ and that $\zz$ satisfies the coupling constraint. 
\smallskip

Recall that we are only interested in \emph{symmetric} action profiles in population games $\GA'$ and $\tG\esnu$. Given a symmetric action profile $\xx\esnu$ in the auxiliary game $\FX\esnu$ in $\tG\esnu$, we can define a corresponding symmetric action profile of $\GA'$ such that all the nonatomic players in the populations regrouped in $\I_n$ play the same action $\xx\esnu_n$. (It is allowed that  $\xx\esnu_n$ be not in $\X_i$. 
Recall that we can extend $u_i$ to a neighborhood of $\M/I$ such that $u_i$  is bounded on $\M/I$). Formally, define map $\psi: \rit^{NT} \rightarrow \rit^{IT}$:
\begin{equation*}
\forall \xx\esnu \hh\hh \in\hh \rit^{N\!T},\, \psi(\xx\esnu) \hh=\hh(\xx_i)_{i\in \I} \, \text{ where } \xx_i \hh= \hh\xx\esnu_n \ , \ \forall i \in \I_n \, .
\end{equation*}
 Conversely, for a symmetric action profile $\xx$ in $\GA'$, we define a corresponding symmetric action profile in the auxiliary game $\tG\esnu$ by the following map $\bpsi: \rit^{IT} \rightarrow \rit^{NT}$:
\begin{align*}
 \forall \xx \hh\in\hh \rit^{IT}, \,\bpsi(\xx)\hh=\hh(\xx\esnu_n)_{n\in \N}\, \text{ where } \xx\esnu_n \hh=\hh \tfrac{1}{I_n}\txt\sum_{i\in \I\esnu_n } \xx_i.
\end{align*}

\Cref{th:main} below is the main result of this subsection. It gives an upper bound on the distance between the SVWE of the population game $\GA'$, which has the same dimension as the original atomic game $\GA$, and that of an auxiliary game $\tG\esnu(A)$, which has a reduced dimension.

\begin{theorem}[SVWE of $\tG\esnu(A)$ is close to SVWE of $\G(A)'$]\label{th:main}
Under \Cref{assp_convex_costs,assp:Xnandun}, in an auxiliary game $\tG\esnu(A)$, $\mdset$ and $\mduti$ are defined by \Cref{eq:def_dset,eq:def_dut}, with $\mdset <\frac{\rho}{2}$. Let $\hxx$ be a SVWE of $\tG\esnu(A)$, and $\sxx$ a SVWE of $\G(A)$. Then: \\
(1) if $H'$ is strongly monotone with modulus $\stgccvut$, then both $\hxx$ and $\sxx$ are unique and
 \begin{equation} \label{eq:cvg_indiv}
\|\psi(\hxx)-\sxx\|^2 \leq \tfrac{1}{ \stgccvut} K\left(\mdset,\mduti\right)\, ; 
\end{equation}
(2) if $H'$ is aggregatively strongly monotone with modulus $\beta$,  then both $\hxxag = \sum_{n\in \N}\hxx_n$ and $\sxxag=\sum_{i\in \I} \sxx_i$ are unique,  and
\begin{equation} \label{eq:cvg_agg_nou}
\| \hxxag- \sxxag \|^2 \leq  \tfrac{1}{\beta}   K\left(\mdset,\mduti\right) \ ,
\end{equation}
where $K(\mdset ,\mduti)$, appearing in both inequalities, is:
\renewcommand{\hh}{\hspace{-0.5pt} }
\begin{equation} \label{eq:KboundExpression}
K(\mdset ,\mduti)\eqd 2M\big( 3 \tfrac{\Bdf  }{\rho}\mdset + \mduti\big) ,
\end{equation}
with $\Bdf \eqd \Bc+\max_{i\in \I}\Bcuti$. In particular, 
\begin{equation*}
K(\mdset \hh\hh ,\mduti)= \mathcal{O}(\mdset + \mduti) \underset{\mdset,\mduti \rightarrow 0}{ \longrightarrow} 0 \ .
\end{equation*}
\end{theorem}
\begin{proof}  See \Cref{app:proof:main}.\end{proof}

We have pointed out that the approximation error depends on how the populations are clustered according to $\N$, and is related to the heterogeneity of players in $\I$ rather than their number. In particular, in the case of \Cref{ex:homogenousPops}, \Cref{th:main} states that the (aggregate) SVWE of the auxiliary game $\tG\esnu(A)$ is exactly equal to the (aggregate) SVWE of the large game $\G(A)'$.

A direct corollary of \Cref{th:main}-(1) is that two populations in $\GA'$ with similar characteristics have similar behavior at a SVWE there. This is the extension of \Cref{prop:continuiteNE} in  the presence of coupling constraints.
\begin{corollary}\label{cor:ij}
Let $\sxx\in \FX$ be a SVWE of game $\GA'$. Under \Cref{assp_convex_costs}, for two populations $i$ and $j$ in $\I$, if $d_H(\X_\i, \X_{j}) \leq \dset$, $\sup_{\xx \in \M/I } d_H (\partial (-u_j)(\xx), \partial (-u_i)(\xx))\leq \duti$,
and $u_i$ (resp. $u_j$) is $\stgccvut_i$- (resp. $\stgccvut_j$-)strongly concave, then 
\begin{align*}
\|\sxx_i-\sxx_j\|\leq & \left( \txt\frac{1}{ \sqrt{\stgccvut_i}}+ \frac{1}{ \sqrt{\stgccvut_j}}\right) \ K\left(\dset,\duti\right) ^{1/2}\ .
\end{align*}
\end{corollary}

\subsection{Combining the two steps to approximate a VNE of $\GA$}

The following theorem is the main result of the paper, which combines the two steps of approximation given in \Cref{th:pseudoVNE} and in \Cref{th:main}, in the computation of a VNE of the original game $\GA$.

\begin{theorem}[SVWE of $\tG\esnu(A)$ is close to VNEs  of $\G(A)$]\label{th:bound}
Under \Cref{assp_convex_costs,assp:Xnandun},  in an auxiliary game $\tG\esnu(A)$, $\mdset$ and $\mduti$ are defined by \Cref{eq:def_dset,eq:def_dut}, with $\mdset <\frac{\rho}{2}$. Let $\hxx$ be a SVWE of $\tG\esnu(A)$, $\xx\in \FX(A)$ be a VNE of  $\GA$, $\hxxag = \sum_{n\in \N}\xx_n$, $\xxag=\sum_{i\in \I} \xx_i$, and $K(\mdset,\mduti) $ the constant given by \eqref{eq:KboundExpression}.\\
(1) if $u_i$ is $\stgccvut_i$-strongly concave for each $i\in \I$, with $\stgccvut \eqd  \min_{i}\stgccvut_i $, then $\hxx$ is unique and
 \begin{equation*}
\|\psi(\hxx)-\xx\|  \leq \sqrt{\tfrac{1}{ \stgccvut}} K\left(\mdset,\mduti\right) ^{1/2}\hh  \+M\sqrt{\tfrac{2TC}{\stgccvut I}}\, ;
\end{equation*}
(2) if  $c_t$ is $\beta_t$-strictly increasing for each $t\in \T$, with $\beta\eqd \min_{t} \beta_t$, then $\hxxag$ is unique and
\begin{equation*}
\| \hxxag- \xxag \|  \leq \sqrt{\tfrac{1}{\beta}}  K\left(\mdset,\mduti\right) ^{1/2} \+ M \sqrt{\tfrac{2TC}{\beta I}}\, .
\end{equation*} 
\end{theorem}
\proof:  This is an implication of the inequalities given in \Cref{th:pseudoVNE,th:main}.
\smallskip

Given the large game $\G(A)$ and a certain $N\in \nit^*$, \Cref{th:bound} suggests that we should find the auxiliary game $\tG\esnu$ that minimizes $K(\mdset,\mduti) $ in order to have the best possible approximation of the equilibria. This would correspond to a ``clustering problem'' given as follows:
\begin{equation} \label{eq:clusteringMinimizeK}
\min_{(\I_n)_n \in \mathcal{P}(\I) }\  \min_{(\X_n)_{n}} \ \min_{(u_n)_{n}} K(\mdset,\mduti) ,
\end{equation}
where $\mathcal{P}(\I)$ denotes the set of partitions of $\I$ of cardinal $N$, while $(\X_n)_{n\in\N}$ and $(u_n)_n $ are chosen according to \Cref{assp:Xnandun}.

The value of the optimal solutions of problem \eqref{eq:clusteringMinimizeK}, and thus of the quality of the approximation in \Cref{th:bound}, depends on the homogeneity of the $I$ players in $\I$ in terms of action sets and utility functions. The ``ideal'' case is given in \Cref{ex:homogenousPops} where $\I$ is composed of a small number $N$ of homogeneous populations and thus $K(\mdset,\mduti)=0$.

In general, solving \eqref{eq:clusteringMinimizeK} is a hard problem in itself.
 It is indeed a generalization of the $k$-means clustering problem \cite{lloyd1982least}  (with $k=N$ and considering a function of Hausdorff distances), which is itself NP-hard \cite{garey1982complexity}.
 In  \Cref{sec:appli}, we illustrate how we  use directly the $k$-means algorithm to compute efficiently an approximate solution $(\I_n, \X_n, u_n)_{n\in\N}$ in the parametric case.



\smallskip
Finally, the number $N$ in the definition of the auxiliary game should be chosen a priori  as a trade-off between the minimization of $K(\mdset,\mduti)$ and a sufficient minimization of the dimension. 
  Indeed, with $\N = \I$, $\X_i=\X_n$ and $u_n=u_i$, we get $\mduti=\mdset=0$. 
However, the aim of \Cref{th:bound} is to find an auxiliary game $\tG\esnu$ with $N \ll I$ so that the dimension of the GVIs characterizing the equilibria (and thus the time needed to compute their solutions) is significantly reduced, while ensuring a relatively small error, measured by $\mduti$ and $\mdset$.

\section{Application to demand response for Electric Vehicle smart charging}
\label{sec:appli}
\newcommand{\ux}{\underline{\x}}
\newcommand{\ox}{\overline{\x}}

Demand response (DR) \cite{ipakchi2009} refers to a set of techniques to influence, control or optimize the electric consumption of agents in order to provide some services to the grid, e.g. reduce production costs  and CO$_2$ emissions or avoid congestion \cite{PaulinTSG17}. The increasing number of  electric vehicles (EV) offers a new source of flexibility in the optimization of the production and demand, as electric vehicles require a huge amount of energy and enjoy a sufficiently flexible charging scheme (whenever the EV is parked).
 Because of the privacy  of each consumer or EV owner's information and the decentralized aspects of the DR problem, many relevant  works adopt a game theoretical approach by considering consumers as players minimizing a cost function and a utility \cite{saad2012game}.

In this section, we consider the consumption associated to electric vehicle charging on a set of 24-hour time-periods  $\T=\{1,\dots,T\}$, with $T=24$, indexing the hours from 10 \textsc{pm} to $9$\textsc{pm} the day after (including the night time periods where EVs are usually parked at home).

\subsection{Price functions: block rates energy prices}

As in the framework described in \cite{PaulinTSG17}, we consider a centralized entity, called the aggregator, who manages the aggregate flexible consumption. The aggregator interacts with the electricity market and energy producers, with his own objectives such as minimizing his cost or achieving a target aggregate demand profile.

The aggregator imposes electricity prices on each time-period. We consider prices taking the specific form of inclining block-rates tariffs (IBR tariffs, \cite{wang2017optimal}), i.e. a piece-wise affine function $c(.)$ which depends on the aggregate-demand $X_t=\sum_{i\in\I} \x\iti$ for each time-period $t$
, and is defined as follows:
\begin{equation}\label{eq:price}
\begin{split}
c(X)&= 1 + 0.1 X  \text{ if } X \leq 500 \\
  c(X)&=-49 + 0.2 X \text{ if } 500 \leq X \leq 1000 \\
  c(X) &=-349 + 0.5 X \text{ if } 1000 \leq X    \ .
\end{split}
\end{equation}
This function $c$ is continuous and convex. Those price functions are transmitted by the aggregator to each  consumer or EV owner. Thus, each consumer $i$ minimizes an objective function of the form \eqref{eq:cost_player_def}, with an energy cost determined by \eqref{eq:price} and a utility function $u_i$ defined below. An equilibrium gives a stable situation where each consumer minimizes her objective and has no interest to deviate from her current consumption profile.

\subsection{Consumers' constraints and parameters} \label{subsec:consumerParams}

\newcommand{\uux}{\bm{\ux}}
\newcommand{\oox}{\bm{\ox}}
We simulate the consumption of $I=2000$ consumers who have demand constraints of the form:\begin{equation}
\label{eq:set_Edemand}
\X_\i\hh=\hh\{\ \xx_\i \in \rit^T_+  :  \txt\sum_\t  \x\iti\hh =\hh m_i \text{ and }  \ux\iti \leq \x\iti \leq \ox\iti\}
\end{equation} 
where $m_i$ is the total energy needed by $i$, and $\ux\iti ,\ox\iti$ the (physical) bounds on the power allowed to her at time $t$.  
The utility functions have form $u_i(\xx_i)=- \omega_i \norm{\xx_i-\yy_i }^{2}$. 

The parameters are chosen as follows:
\begin{itemize}[wide]
\item $m_i$ is drawn uniformly between 1 and 30 kWh, which corresponds to a typical charge of a residential electric vehicle.
\item $\uux_i, \oox_i$:
 First, we generate, in two steps,  a continual set  of charging time-periods $\T_i=
\{ h_i - \frac{\tau_i}{2} , \dots , h_i + \frac{\tau_i}{2} \}$: 
\begin{itemize}
\item the duration $\tau_i$ is uniformly drawn from $ \{4,\dots,T \}$;
\item $h_i$ is then uniformly drawn from $\{ 1+\frac{\tau_i}{2} , \dots , T- \frac{\tau_i}{2}\}$. 
\end{itemize}
Next, for $t\notin \T_i$, let $\ux\iti=\ox\iti=0$.\\
Finally, for $t\in \T_i$, $\ux\iti$ (resp.  $\ox\iti$) is drawn uniformly from $[0,\frac{m_i}{\tau_i}]$ (resp. $[\frac{m_i}{\tau_i}, m_i]$).
\item $\omega_i$ is drawn uniformly from $[1,10]$.
\item $\yy\iti$ is taken equal to $\ox\iti$ on the first time periods of $\T_i$ (first available time periods) until reaching $m_i$ (which corresponds to a profile ``the sooner the better" or ``plug and charge'').
\end{itemize}

\subsection{Coupling constraints on capacities and limited variations}

We consider the following \emph{coupling} constraints  on the aggregate demand $\xxag$ which are often encountered in energy applications:
\begin{align}
\label{cst:ramp}
-50 \leq \xag_{T}-\xag_1 \leq 50 \\
\label{cst:capa} \xag_t \leq 1400 , \quad \forall \t \in \T
\end{align}

Here, Constraint \eqref{cst:ramp} imposes that the demand $X_T$ at the very end of the time horizon is relatively close to the first aggregate $X_1$, so that the demand response profiles computed for the finite time set $\T$ can be applied on a day-to-day, periodical basis.

Constraint \eqref{cst:capa} is a capacity constraint, 
 induced by the maximal capacity of the electrical lines or by the generation capacities of electricity producers.

These linear coupling constraints can be written in the closed form:
\begin{equation}
\label{eq:coupling_linear}
A \xxag \leq b \ , 
\end{equation}
where $A \in \M_{T+2,T}(\rit)$, $b\in \rit^{T+2}$.

\subsection{Computing populations with $k$-means}
\newcommand{\pp}{\bm{p}}
Since $I$ is very large, determining an exact VNE is computationally  demanding. Thus, we apply the clustering procedure described in \Cref{subsec:class} to regroup the players.

We use the $k$-means algorithm \cite{lloyd1982least}, where ``$k$''$=N$ is the number of populations (groups) to replace the large set of $I$ players. For each player $i\in \I$, we define her parametric description vector:
\begin{equation}
\pp_i = [ \omega_i, \yy_i, m_i, \uux_i, \oox_i ] \in \rit^{3T+2}\ .
\end{equation}
Then, the $k$-means algorithm finds an approximate solution of finding a partition $(S_n)_{1\leq n \leq N}$ of $\I$ into $N$ clusters. The algorithm solves the combinatorial minimization problem:
\begin{equation*}
\hh \min_{S_1,\dots, S_N} \hspace{-4pt} \sum_{1\leq n \leq N} \sum_{\pp \in S_n} \hh \norm{ \mathbb{E}_{S_n}\hspace{-2pt}(\pp) - \pp   }^2 \hh = \hspace{-5pt}\min_{S_1,\dots, S_N}  \hh \hspace{-3pt} \sum_{1\leq n \leq N} \hspace{-6pt} |S_n|  \mathbb{V}\text{ar}(S_n) ,
\end{equation*}
where $\mathbb{E}_{S_n}(\pp)=\frac{1}{|S_n|}\sum_{i\in S_n} \pp_i$ denotes the average value of $\pp$ over the set $S_n$. These average values are taken to be $w_n$, $\yy_n$, $m_n$, $\uux_n$ and $\oox_n$.


The simulations are run with different population numbers, with $N\ll I$  chosen among $\{5,10,20,50,100\}$.

\medskip 
Since the $k$-means algorithm minimizes the squared distance of the average vector of parameters in $S_n$ to the vectors of parameters of the points in $S_n$, the clustered populations obtained can be sub-optimal in terms of  $K(\mdset, \mduti)$. As explained above,  choosing the optimal populations $\N$, as formulated in problem $\eqref{eq:clusteringMinimizeK}$, is a complex problem in itself which deserves further research. Our example shows that the $k$-means algorithm gives a practical and efficient way to compute a heuristic solution in the case where $u_i$ and $\X_i$ are parameterized.

\subsection{Computation methods}

We compute a VNE (\Cref{def:ve-finite}) with the original set of $I$ players and the approximating SVWE (\Cref{def:pseudoVNE}) as solutions of the associated GVI \eqref{cond:ind_opt_ve}.

We employ a standard projected descent algorithm, as described in \cite[Algo. 2]{gentile2017nash} and recalled below in \Cref{algo:PGDFS}. It is adapted  to the subdifferentiable case that we consider in this work. In particular, the fixed step $\tau$ used in \cite{gentile2017nash} is replaced by a variable step $\tau^{(k)}=1/k$. The coupling constraint \eqref{eq:coupling_linear} is relaxed and the Lagrangian multipliers $\bm{\lambda} \in \rit_+^{T+2}$ associated to these constraints are considered as extra variables. Thus, we can perform the projections on the sets $\X_i$ and on $\rit_+^{T+2}$.
\begin{algorithm}[H]
\begin{algorithmic}[1] 
\Require ${\xx}^{(0)}, \bm{\lambda}^{(0)}$, stopping criterion 
\State $k \leftarrow 0$  \;
\While
{stopping criterion not true}
\For{$n=1$ to $N$} \label{line:algo-for}
\State take $g_n^{(k)} \in \partial_1 f_n(\xx_n^{(k)},\xxag^{(k)}) $ \; 
\State $\xx_n^{(k+1)} \leftarrow \Pi_{\X_n}\left(\xx_n^{(k)}-\tau^{(k)}(g_n^{(k)} + {\bm{\lambda}^{(k)}}^T A ) \right) $ 
\label{line:algo-proj}\;
\EndFor
\State $\bm{\lambda}^{(k+1)} \leftarrow \left( {\bm{\lambda}^{(k)} - \tau^{(k)} (b-2 A\xxag^{(k+1)} + A \xxag^{(k)} } \right)^+$ \;
\State $  k \leftarrow k+1 $ \;
\EndWhile
\end{algorithmic}
\caption{Projected Descent Algorithm }
\label{algo:PGDFS}
\end{algorithm}

The convergence of \Cref{algo:PGDFS} is shown in \cite[Thm.3.1]{cohen1988auxiliary}.
The stopping criterion that we adopt here is the distance between two iterates: the algorithm stops when $\norm{ (\lambda^{(k+1)},\xx^{(k+1)})-(\lambda^{(k)},\xx^{(k)}) }_2 \leq 10^{-3}$.

Due to the form of the strategy sets considered \eqref{eq:set_Edemand}, the projection steps (Line 5) can be computed efficiently and exactly in $\mathcal{O}(T)$ with the Brucker algorithm \cite{brucker1984n}. However, if we consider more general strategy sets (arbitrary convex sets), this projection step can be costly: in that case, other algorithms such as \cite{fukushima1986relaxed} would be more efficient.                    

\subsection{A trade-off between precision and computation time}

\begin{figure} 
\hspace{-13pt}
\includegraphics[width=1.02\columnwidth]{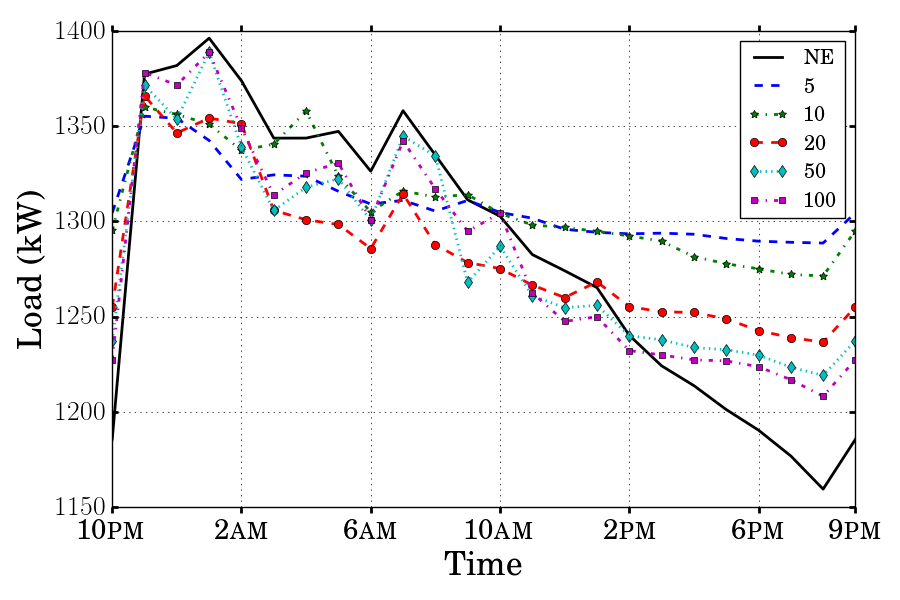}
\caption{Convergence of the aggregate SVWE profile of auxiliary games to a VNE profile of the original game. 
}
\label{fig:cvg_pne_pop}
\end{figure}

Simulations were run using Python on a single core Intel Xeon @3.4Ghz and 16GB of RAM.

\Cref{fig:cvg_pne_pop} shows the different aggregate SVWE profiles $\hxxag\esnu$ obtained for sets $\snu$ of different sizes, as well as a VNE of the original game for comparison. 
Thanks to the specific form of the strategy sets \eqref{eq:set_Edemand}---which enables a fast projection---we are able to compute a VNE $\xxag$ of the original game with $I=2000$ players.

\hspace{-5pt}\begin{figure}
     \centering
     \subfloat[][]{\includegraphics[width=0.485\columnwidth]{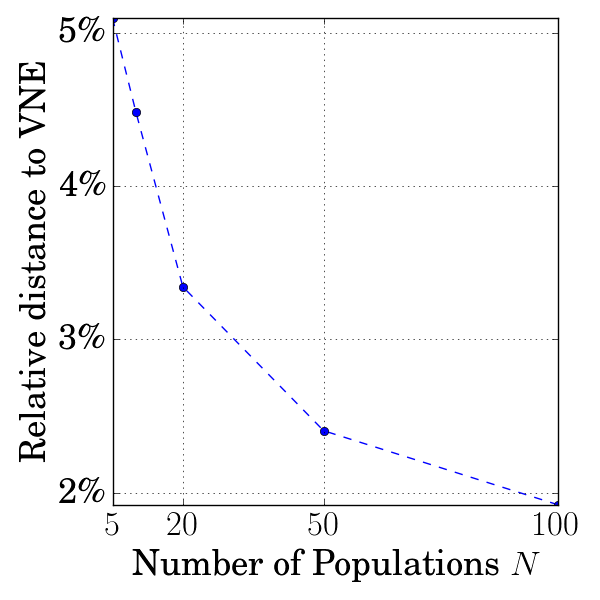} \label{fig:cvg_pne_dists} }
     \subfloat[][]{\includegraphics[width=0.485\columnwidth]{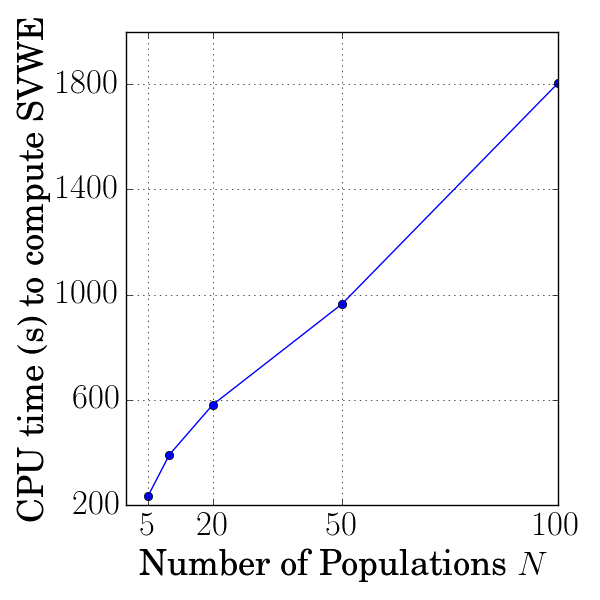}
     \label{fig:cvg_pne_CPU}}
     \caption{(a) Relative error to actual VNE ; (b) Time to compute  SVWE. \textit{ The time needed to compute SVWE (b) increases roughly linearly with $N$, at a faster rate than the error on the approximation of the VNE decreases (a). }
     }
\end{figure}

%

\Cref{fig:cvg_pne_dists} and \Cref{fig:cvg_pne_CPU} show  the two main metrics  to consider to choose a relevant number of populations $N$: the precision of the SVWE approximating the equilibrium (measured by the distance of the  aggregate SVWE profile to the aggregate profile of the VNE computed along), and the CPU time needed to compute the SVWE.

First notice on \Cref{fig:cvg_pne_dists} that the distance between the aggregate equilibrium profile and its estimation decreases with $N$ at a sublinear rate. 
This is partially explained in light of \Cref{th:main} and in addition with the following remarks:
\begin{itemize}[wide]
\item the Hausdorff distance of two parameterized polyhedral sets is Lipschitz continuous w.r.t their parameter vectors (generalization of \cite{batson1987combinatorial}), which ensures that there is $K>0$ s.t.  for all $n$:
\vspace{-0.4cm}
\begin{equation*}
\dset\esnu_n = \max_{\i \in \I\esnu_n} d_{H}\left( \X_\i, \X_{n}\esnu \right) \leq K \norm{ \begin{pmatrix}
m_n \\ \ux_n \\ \ox_n 
\end{pmatrix}-\begin{pmatrix}
m_i \\ \ux_i \\ \ox_i
\end{pmatrix}  }  \ ;
\end{equation*}\
\item similarly, as subgradients of utility functions are reduced to a point, one has, for all $n$:
\begin{align*}
\duti_n &= \max_{i\in } \max_{\xx \in \M} 2 \norm{ \omega_n (\xx  - \yy_n)- \omega_i (\xx-\yy_i)} \\
& = \mathcal{O}\left( | \omega_n-\omega_i| + \norm{ \yy_n- \yy_i} \right) \ .
\end{align*}
\end{itemize}

\Cref{fig:cvg_pne_CPU} shows the CPU time needed to compute the WE with a stopping criterion of a maximum improvement between iterates of $\norm{ (\lambda^{(k+1)},\xx^{(k+1)})-(\lambda^{(k)},\xx^{(k)}) }_2 \leq 10^{-3}$. 
Computing a solution of the clustering problem with the $k$-means algorithm takes, for each value of $N$, less than ten seconds. This time is negligible in comparison to the time needed for convergence of \Cref{algo:PGDFS}.

As a reference time,  to compute a VNE of the original game (observed on \Cref{fig:cvg_pne_pop}) with the same stopping criterion and the same CPU configuration, we needed 3 hours 26 minutes. This is more than six times longer than the CPU time to compute the SVWE with one hundred populations.

On this figure, we see that the CPU time evolves linearly with the number of populations $N$. This  is  explained by the structure of \Cref{algo:PGDFS}, as each iteration $k$  is executed in a time proportional to $N$ due to the \textbf{for} loop. 


Last, one observes from \Cref{fig:cvg_pne_dists} that, in our example, the error between the aggregate demand profile at equilibrium  and its approximation is between 2\% and 5\%, which remains significant. However, as pointed out in \Cref{sec:approxRes}, the quality of the approximation depends on the heterogeneity of the set of players $\I$. In the example of this section, as the parameters are drawn uniformly (see \Cref{subsec:consumerParams}), the set of players $\I$ presents a large variance so that it is a ``worst'' case as opposed to the case of \Cref{ex:homogenousPops} which is ``optimal''.

\section{Conclusion}

This paper shows that equilibria in splittable congestion games with a very large number of atomic players can be approximately computed with a Wardrop equilibrium of an auxiliary population game of smaller dimension. 
Our results give explicit bounds on the distance of this approximating equilibrium to the equilibria of the original large game. 
These theoretical results can be used in practice to solve, by an iterative method, complex nonconvex bi-level programs where the lower level is the equilibrium of a large congestion game, for instance, to optimize tariffs or tolls for the operator of a network.
A detailed analysis of such a procedure would be an extension of the present work.

\section{Acknowledgments}

We thank the PGMO foundation for the financial support to our project \textit{``Jeux de pilotage de flexibilit\'es de consommation \'electrique : dynamique et aspect composite".}

\appendices

\crefalias{section}{appsec}

\section{Proof of \Cref{prop:subgradients-sets}: Expressions of Subgradients}
\label{app:propSubgradients}

Recall that $\hat{f}_i(\xx_i,\xxag_{-i}) \eqd f_i(\xx_i, \xxag_{-i}+\xx_i)$. 
According to \cite[Proposition 16.6]{combettes2011monotone}, $\partial_1 \hf_\i(\xx_\i, \xxag_{-i})\subset \{(I_T,I_T) \g :\g\in \partial \psi_\i(\xx_\i)\}$, where $\partial f_\i(\xx_\i, \xxag_{-i}+\xx_i)$ is the subdifferential of $\psi_i(\cdot)\eqd f_\i(\cdot, \xxag_{-i}+\cdot)$, at $\xx_\i$. On the other hand, according to \cite[Proposition 16.7]{combettes2011monotone}, $\partial \psi_\i(\xx_\i)$ is a subset of:
\begin{equation*}
\begin{split}
\big\{ 
(\g_{i,1} , \g_{i,2} ): 
& \g_{i,1} \in\partial_1 f_\i(\ww_\i, \yyag)|_{\ww_i=\xx_i,\yyag=\xxag_{-i}+\xx_i}, \, \\ 
& \g_{i,2} \in \partial_2 f_\i(\ww_\i, \yyag)|_{\ww_i=\xx_i,\yyag=\xxag_{-i}+\xx_i}\big\} \ .
\end{split}
\end{equation*} 
 Therefore, $\partial_1 \hf_\i(\xx_\i, \xxag_{-i})$ is a subset of:
 \begin{equation*}
\begin{split}
& \{\cc(\xxag) \+\g'_{i,1}\hh + \g_{i,2} \hh:  \g'_{i,1} \in \partial (-u_i)(\xx_i), \, 
\g_{i,2} \in \partial_2 f_\i(\xx_\i, \xxag)\} \\ 
& =   \{\cc(\xxag)+\g'_{i,1} + (a_{i,t}x_{i,t})_t:   \\   & \ \ \ \ \ \ \g'_{i,1} \in \partial (-u_i)(\xx_i),  \, a_{i,t} \in \partial c_\t(\xag_t) \,\forall t\in T\} \ . \
\end{split}
\end{equation*}
 By the definition of subdifferential, it is easy to show that $ \{\cc(\xxag)+\g'_{i,1} + (a_{i,t}x_{i,t})_t:   \g'_{i,1} \in \partial (-u_i)(\xx_i), \, a_{i,t} \in \partial c_\t(\xag_t) \,\forall t\in T\} \subset \partial_1 \hf_\i(\xx_\i, \xxag_{-i})$.

The proof for $\partial_1 f_\i(\xx_\i, \xxag)$ is similar.

\section{Proof of \Cref{prop:exist_ve}: Existence of equilibria}
\label{app:proof:exist_ve}
It is easy to see that $\hf_i(\cdot,\xxag_{-i})$ is  convex on $\X_i$ for all $\xxag_{-i}$ in $\sum_{j\in \I\setminus\{i\}} \X_j$,  $f_i(\cdot,\xxag)$ is convex on $\X_i$ for all $\xxag\in \Sxag$, and $H$ and $H'$ are nonempty, convex, compact valued, upper hemicontinuous correspondences. Then \cite[Corollary 3.1]{chanpang1982gqvip} shows that the GVI problems \eqref{cond:ind_opt_ve} and \eqref{eq:def-pseudo} both admit a solution on the finite dimensional convex compact $\FX(A)$. 

\section{Proof of \Cref{th:unique_vwe}: uniqueness of equilibria}
\label{app:proof-uniquenessVWE}
We prove for SVWE only and the proof for VNE is the same. Suppose that $\xx,\yy\in \FX(A)$ are both SVWE, with $\xxag=\sum_i\xx_i$ and $\yyag=\sum_i \yy_i$. According to the definition of SVWE, there is $\g\in H'(\xx)$ an $\bh\in H'(\yy)$ such that $ \sum_i \langle \g_\i, \yy_\i - \xx_\i \rangle \geq 0$ and $\sum_{i} \langle \bh_\i, \xx_\i - \yy_\i \rangle \geq 0$. Adding up these two inequalities yields:
\begin{equation*}
\txt\sum_{\i}\langle \g_\i - \bh_\i, \yy_\i - \xx_\i \rangle \geq 0. 
\end{equation*}
(1) If $H'$ is a strictly monotone, then $\sum_{\i}\langle  \g_\i - \bh_\i, \xx_\i - \yy_\i \rangle = 0$ and thus $\xx=\yy$.  \\
(2-3) If $H'$ is an aggregatively strictly monotone, then $\sum_{\i}\langle  \g_\i - \bh_\i, \xx_\i - \yy_\i \rangle = 0$ and thus $\xxag=\yyag$.  If there is no aggregative constraint and $u_\i$ is strictly concave, then $\xx_\i$ (resp. $\yy_\i$) is the unique minimizer of $f_\i(\cdot, \xxag)$ (resp. $f_\i(\cdot, \yyag)$). Since $\xxag=\yyag$, one has $\xx_\i=\yy_\i$.

\section{Proof of \Cref{prop:monotonemap}: monotonicity of $H'$}
\label{app:proof-monotonicity}
\noindent (1) Let $\xx,\yy\in \FX$ and $\xxag=\sum_i \xx_i$, $\yyag=\sum_i \yy_i$. Recall that
\begin{align*}
& \partial_1 f_i(\xx_i, \xxag)=\{\cc(\xxag)+\g:\g\in \partial (-u_i(\xx_i) \} \ , \\
& \partial_1 f_i(\yy_i, \yyag)=\{\cc(\yyag)+\bh:\bh\in \partial (-u_i)(\yy_i) \}  \ . 
\end{align*}
Let $\g_i \in  \partial (-u_i)(\xx_i)$ and $\bh_i \in  \partial (-u_i)(\yy_i)$. One has $\langle  \g_i(\xx_i)-\bh_i (\yy_i), \xx_i - \yy_i \rangle\geq 0$ because $u_i$ is concave so that $\partial (-u_i)$ is monotone on $\X_i$. Then we get:
\begin{align*}
& \txt\sum_i \langle (\cc(\xxag)+\g_i(\xx_i))-(\cc(\yyag)-\bh_i(\yy_i)), \xx_i - \yy_i \rangle  \\ 
= & 
\langle \cc(\xxag)-\cc(\yyag), \xxag-\yyag\rangle + \txt\sum_i \langle  \g_i (\xx_i) - \bh_i (\yy_i), \xx_i - \yy_i \rangle \\ 
  \geq  &  0
\end{align*}
 because $\cc$ is monotone. Hence $H'$ is monotone.
 
\noindent (2)  By the definition of $\stgccvut_i$-strong concavity:
\begin{align*}
& \txt\sum_i \langle  \g_i (\xx_i) \hh- \hh\bh_i (\yy_i), \xx_i\hh -\hh \yy_i \rangle \hh \\ 
 & \geq \txt\sum_i \stgccvut_i \|\xx_i\hh-\hh\yy_i\|^2 \hh\geq \hh\stgccvut \|\xx\hh-\hh\yy\|^2 \ .
\end{align*}

 \medskip
\noindent (3) By the definition of $\beta_t$-strong monotonicity:
\begin{align*}
&\langle \cc(\xxag)-\cc(\yyag), \xxag-\yyag\rangle\\
& = \txt\sum_{t\in \T} \langle c_t(\xag_t) - c_t(\yag_t), \xag_t - \yag_t \rangle \\
& \geq \txt\sum_t\beta_t \|\xag_t- \yag_t\|^2 \geq \beta\|\xxag-\yyag\|^2 \ .
\end{align*}

\section{Proof of \Cref{th:VNEareclose} : VNEs are close to each other}
\label{app:proof-VNEclose}

%

(1) Let $\xx,\yy\in \FX$ be two VNEs. Then, by \eqref{cond:ind_opt_ve}, there are $\g_i \in \partial_1 \hf_i(\xx_i,\xxag_{-i})$ and $\bh_i \in \partial_1 \hf_i(\yy_i,\yyag_{-i})$ for each $i$  with $\g_i = \cc(\xxag)+\g'_i +(a\iti x\iti)_{t}$, $\bh_i=\cc(\yyag)+\bh'_i+(b\iti y\iti)_{t}$, 
where $\g'_i\in \partial (-u_i)(\xx_i)$, $\bh'_i\in \partial (-u_i)(\yy_i)$, $a\iti \in \partial c_t(\xag_t)$ and $b\iti \in \partial c_t(\yag_t)$ for all $t$, such that $\sum_i \big\langle\g_i, \yy_i - \xx_i \big\rangle\geq 0$ and $\sum_i \big\langle\bh_i, \xx_i- \yy_i \big\rangle\geq 0$.

Summing up these two inequalities yields:
\renewcommand{\-}{\hspace{-2pt}-\hspace{-2pt}}
\begin{align*}
0&\leq   \txt\sum_i \langle \g_i -\bh_i, \yy_i - \xx_i \rangle  \\
= &   \txt\sum_i\! \langle\cc(\xxag\!)\+\g'_i \+(a\iti x\iti)_{t} \-\cc(\yyag\!)\-\bh'_i \-(b\iti y\iti)_{t}, \yy_i \-\xx_i \rangle \\
=&  \txt\sum_i \langle\cc(\xxag) -\cc(\yyag), \yy_i -\xx_i \rangle +\sum_i \langle\g'_i -\bh'_i , \yy_i -\xx_i \rangle \\ 
&+ \txt\sum_i \langle(a\iti x\iti)_{t}-(b\iti y\iti)_{t}, \yy_i -\xx_i \rangle \\
=& \langle\cc(\xxag) -\cc(\yyag), \yyag -\xxag \rangle + \txt\sum_i \langle\g'_i -\bh'_i , \yy_i -\xx_i \rangle \\
&+ \txt\sum_i \langle(a\iti x\iti)_{t}-(b\iti y\iti)_{t}, \yy_i -\xx_i \rangle \ .
\end{align*}
Therefore,
\begin{align*}
&\langle\cc(\xxag) -\cc(\yyag), \xxag -\yyag \rangle +\txt\sum_i \langle\g'_i -\bh'_i , \xx_i -\yy_i \rangle\\
&  \leq - \txt\sum_{i}\sum_{t} (a\iti x\iti - b\iti y\iti)(x\iti - y\iti) \\
&\leq \txt\sum_{i,t} ( 2C\tfrac{M}{I} )(2\tfrac{ M}{I})
 = \tfrac{4TC M^2}{I} \ .
\end{align*}
Since $\cc$ is monotone and so are $\partial u_i$'s because $u_i$'s are concave, $\langle\cc(\xxag) -\cc(\yyag), \xxag -\yyag \rangle\geq 0$, $\sum_i \langle\g'_i -\bh'_i , \xx_i -\yy_i \rangle \geq 0$.

If for each $i$, $u_i$'s are $\stgccvut_i$-strongly concave, then $ \stgccvut \sum_i  \|\xx_i -\yy_i \|^2 \leq \sum_i \stgccvut_i   \|\xx_i -\yy_i \|^2  \leq \sum_i \langle\g'_i -\bh'_i , \xx_i -\yy_i \rangle \leq  \tfrac{4TC M^2}{I} $
so that $\|\xx-\yy\| \leq 2M\sqrt{\frac{TC}{\stgccvut I}} $.

If for each $t$, $c_t$ is $\beta_t$-strictly increasing, then $\beta \|\xxag-\yyag\|^2\leq\langle\cc(\xxag) -\cc(\yyag), \xxag -\yyag \rangle \leq   \tfrac{4TC M^2}{I}$ thus $\|\xxag - \yyag\| \leq 2M \sqrt{\frac{TC}{\beta I}}$.

\section{Proof of \Cref{prop:continuiteNE}: SWE behavior for similar players}
\label{app:proof:lmcontinuite}
Let $\g'_i(\sxx_i)\in \partial (-u_i)(\sxx_i)$ be s.t., for all $\xx_i\in \X_i$, $ \langle \cc(\sxxag) + \g'_i(\sxx_i), \sxx_i -\xx_i \rangle \leq 0$. 
Let $\bh'_i(\sxx_j) \in \partial (-u_i)(\sxx_j)$ be such that $\|\bh'_i (\sxx_j)- \g'_j (\sxx_j)\|\leq \duti$. 
Then, by the strong concavity of $u_i$:
\begin{align*}
&\stgccvut_i \norm{\sxx_i  -\sxx_j}^2  \leq   \langle \g'_i(\sxx_\i)-\bh'_i(\sxx_j), \sxx_i-\sxx_j \rangle \\
= & \langle \g'_i(\sxx_i)-\g'_j(\sxx_j)+\g'_j(\sxx_j)-\bh'_i(\sxx_j), \sxx_i-\sxx_j \rangle   \\
\leq & \langle \g'_i(\sxx_i) - \g'_j(\sxx_j),\sxx_i-\sxx_j  \rangle+\duti 2m \\
= & \langle \g'_i( \sxx_i)\hh +\hh\cc(\sxxag),\sxx_\i\hh -\hh \sxx_j  \rangle  \\
& + \langle \g'_j(\sxx_j)\hh + \hh \cc(\sxxag) ,\sxx_j \hh -\hh \sxx_i \rangle \hh +\hh 2 \duti m\\
= &  \langle \g'_i(\sxx_i) +\cc(\sxxag),\sxx_\i-\Pi_{i}(\sxx_j)+\Pi_{i}(\sxx_j)-\sxx_j  \rangle \\
 &\hh+\hh \langle \g'_j(\sxx_j) \hh+\hh\cc(\sxxag) ,\sxx_j\hh -\hh\Pi_{j}(\sxx_i)\hh+\hh\Pi_{j}(\sxx_i)\hh-\hh \sxx_i \rangle \hh+\hh 2 \duti m\\
\leq &  \langle \g'_i(\sxx_i) +\cc(\sxxag),\Pi_{i}(\sxx_j)-\sxx_j  \rangle  \\ 
& +\langle \g'_j(\sxx_j) +\cc(\sxxag) ,\Pi_{j}(\sxx_i)- \sxx_i \rangle+ 2\duti m\\
\leq & (\Bcuti+\Bcutj+2\Bc) \delta + 2\duti m\ .
\end{align*}
where $\Pi_i$ (resp. $\Pi_j$) is the projector on $\X_i$ (resp. $\X_j$).

%

\section{Proof of \Cref{lm:intprofile}: Existence of interior profile}
\label{app:proof:lm-intprofile}

Let $\bar{\xx}\in \FX$ be s.t. $d(\bar{\xx}_\i, \rbd \X_\i)=\max_{\xx\in \X_i}d(\xx,\rbd \X_i) \eqd \eta_i$, for all $i$. Denote $\bar{\xxag}=\sum_i \bar{\xx}_i$ and 
$\eta=\min_{i} \eta_i>0$. 

Let $\yy\in\FX(A)$ and $\yyag=\sum_i \yy_i$ be s.t. $d(\yyag, \rbd A)=\max_{\xxag\in \Sxag\cap A}d(\xxag,  \rbd A)$.  Denote $t=\frac{d(\yyag, \rbd A)}{3M}$. 

 Define $\zz =\yy - t(\yy -\bar{ \xx})\in \FX$. Let $\zzag=\sum_i \zz_i$.

Firstly, $\|\yyag- \zzag\|=t\|\yyag-\bar{\xxag}\|\leq t 2M\leq \frac{2}{3}d(\yyag, \rbd A)$, hence $\zzag \in \Sxag\cap \rlt  A$, where $\rlt$ means the relative interior. 
Besides, for any $\i$, $\zz_i =\yy_\i - t(\yy_\i - \bar{\xx}_\i) $. Since $d(\bar{\xx}_\i, \rbd \X_\i)\geq \eta$, $\yy_\i \in \X_\i$, and $\X_\i$ is convex, one has $d(\zz_\i, \rbd \X_\i ) \geq \eta t = \frac{\eta}{3M}d(\yyag, \rbd A)$. Finally, define $\rho \eqd \frac{\eta}{3M}d(\yyag, \rbd A)$. 
%
%

\section{Proof of \Cref{th:main}: approximation of SVWE}
\label{app:proof:main}
%
\begin{lemma}\label{lm:FY}~~\\
(1) For each $n\in \N$ and $\xx\in \X_n$, if $d(\xx, \rbd \X_n)>\dset_n$, then $\xx\in \X_i$ for each $i\in \I_n$.\\
(2) For each $n\in \N$, $i\in \I_n$ and $\xx\in \X_i$, if $d(\xx, \rbd \X_i)>\dset_n$, then $\xx\in \X_n$.
\end{lemma}
\begin{proof}[Proof of \Cref{lm:FY}]
(1) Suppose  $\xx \notin \X_i$. Let $\yy \eqd \Pi_{\X_i}(\xx)\neq \xx$. As $\yy \in \aff \X_i \subset \aff \X_n $, then $\xx-\yy \in \aff \X_n$. Let $\zz \eqd \xx +  \dset_n \frac{ \xx-\yy}{\norm{\xx-\yy}}$. Then, $\zz \in \X_n$  because $\norm{\zz - \xx } \leq \dset_n$.
By the convexity of $\X_i$ and the definition of $\yy$, we have $d(\zz,\X_i)=  d(\xx,\X_i)+ \dset_n> \dset_n$,  contradicting the fact that $\dset_n \geq  d_H(\X_i,\X_n)$. (2) Symmetric proof.
\end{proof}

\begin{lemma}\label{lem:dist_agg_genized_sets}
Under \Cref{assp_convex_costs}, if $\mdset <\frac{\rho}{2}$, then\\
(1) for each $\xx \in \FX\esnu(A)$, there is $\ww\in \FX(A)$ such that $\|\ww_i - \psi_i(\xx)\| \leq 4m \tfrac{\mdset}{\rho}$ for each $i\in \I$;
\\
(2) for each $\xx\in \FX(A)$, there is $\ww\in \FX\esnu(A)$ such that $\|\ww_n-\bpsi_n(\xx)\|\leq 2m I_n\frac{\mdset}{\rho}$ for each $n\in \N$.
\end{lemma}
\begin{proof}[Proof of \Cref{lem:dist_agg_genized_sets}]~\\
(1) For $\xx \in \FX\esnu(A)$, define $\ww\in \FX$ as follows: $\forall n\in\N$, $\forall \i\in \I_n$, let $\ww_i \eqd \xx_n + t(\zz_i - \xx_n )$ where $\zz$ is defined in \Cref{lm:intprofile}, with $t\eqd 2\mdset/\rho<1$. 


On the one hand, $\forall n\in\N$, $\forall \i\in \I_n$, $d(\zz_i, \rbd \X_n) \geq \rho - \mdset$ implies that $d(\ww_i,\rbd \X_n) \geq t(\rho-\mdset)> t\rho/2 =\mdset $.  (This is because each point in the ball with radius $t( \rho-\mdset)$ centered at $\ww_i$ is on the segment linking $\xx_n$ and some point in the ball with radius $\rho-\mdset$ centered at $\zz_i$ which is contained in $\X_i$.) Thus, $\ww_i \in \X_\i$  $\forall \i\in \I_n$ according to \Cref{lm:FY}.(1). 
On the other hand, the linear mapping $S: \rit^{IT} \ni \bv \mapsto \sum_{i\in \I} \bv_i$ maps the segment linking $\psi(\xx)$ and $\zz$ in $\FX(A)$ to a segment linking $\xxag=\sum_n  I_n \xx_n$ and $\zzag$ in the convex $ A$. Hence $\sum_{i\in \I} \ww_i = t\zzag + (1-t)\xxag$ is in $A$ as well. 
Therefore, $\ww\in \FX(A)$. 

Finally, $\|\ww_i - \psi_i(\xx)\| =t\|\zz_i-\psi_i(\xx)\| \leq t2m=  4m \tfrac{\mdset}{\rho}$.

%
%

(2) For $\xx\in \FX(A)$, let $\yy\eqd \xx+t(\zz-\xx)$ with $t\eqd \frac{\mdset}{\rho}$. Then, by similar arguments as above, $d(\yy_i, \rbd \X_i) \geq \mdset$ hence $\yy_i\in \X_n$ and $\bpsi(\yy)\in \FX\esnu$. Besides, $\sum_i \yy_i=t\zzag+(1-t)(\sum_i \xx_i)$ so that $\sum_i \yy_i$ is in the convex $A$. Hence $\ww\eqd \bpsi(\yy)\in \FX\esnu(A)$. Finally, $\|\ww_n-\psi_n(\xx)\|=t\|\sum_{i\in I_n}(\zz_i-\xx_i)\|\leq 2m I_n\frac{\mdset}{\rho}$ .
\end{proof}



Let $\ww\in \FX(A)$ be s.t. $\forall i\in \I$, $\|\ww_i - \psi_i(\hxx)\| \leq 4 m \mdset/\rho$  (cf. \Cref{lem:dist_agg_genized_sets}). Since $\sxx$ is a SVWE in $\GA'$, there is $\g'_i\in \partial (-u_i)(\sxx_i), \forall i\in \I$ s.t. $\sum_{i} \langle \cc(\sxxag)+\g'_i , \ \sxx_\i  - \ww_\i  \rangle \leq 0$. 
Secondly, since $\hxx$ is a SVWE in $\tG(A)$, there is $\bh'_n \in \partial (-u_n)(\hxx_n), \forall n\in\N$ s.t. $\sum_{n} I_n \langle \cc(\hxxag)+ \bh'_n, \hxx_n - \yy_n\rangle \leq 0$ for all $\yy \in \X\esnu(A)$. 
Thirdly,  $\forall n, \forall \i \in \I_n$, by the definition of $\duti_i$, there is $\rr'_i\in \partial (-u_\i)(\hxx_n) $ such that $ \|\rr_i - \bh'_n \|\leq \duti_i$. 

The above results and $\hxx_n \leq m, \forall n$ imply:
\renewcommand{\-}{\hspace{-2pt}-\hspace{-2pt}}
\begin{align}
\nonumber
&\langle\cc(\sxxag)\- \cc(\hxxag),\sxxag\-\hxxag \rangle \+\txt\sum_i \hh  \big\langle  \g'_i\- \rr'_i ,  \sxx_i \-   \hxx_n \big\rangle\\
&= \hh\langle \cc(\sxxag)\hh-\hh\cc(\hxxag), \sxxag \hh- \hh\hxxag\rangle\hh  + \hh\txt\sum_{n,i\in \I_n} \hh\langle   \g'_i\hh- \hh\rr'_i ,  \sxx_i\hh- \hh\hxx_n\rangle \notag \\
&= \hh  \txt\sum_{n,i\in \I_n} \hh\hh\big[\langle \cc(\sxxag)\hh+\hh \g'_i ,  \sxx_i\hh  -\hh \ww_i \rangle \hh + \hh \langle \cc(\sxxag)\hh+\hh \g'_i , \ww_i \hh-\hh \hxx_n  \rangle\big]  \notag  \\
& \quad + \txt\sum_{n,i\in \I_n} \big[ \langle  \rr'_i \hh-\hh\bh'_n  ,  \hxx_n\! -\! \sxx_i  \rangle\hh +\hh\langle \cc(\hxxag)\hh+\hh \bh'_n ,  \hxx_n \!-\! \sxx_i  \rangle \big]  \notag \\
& \leq 0\!+\! \txt\sum_{n,i\in \I_n}  \norm{\cc(\sxxag)+\g'_i} \norm{\ww_i\!-\!\hxx_n}  \! \notag \\ 
& \quad+  \!  \txt\sum_{n,i\in \I_n} \|\rr'_i \!-\! \bh'_n \| \norm{\hxx_n\! -\! \sxx_\i }\! + \! J\notag\\
&  \leq  \Bdf  \, 4M \tfrac{\mdset}{\rho}  +\,2M \mduti+J   \label{eq:th-firstbound}
\end{align}
where $J\eqd  \sum_{n,i\in \I_n} \big\langle  \cc(\hxxag)+\bh'_n , \ \hxx_n - \sxx_i  \big\rangle $.

Next, for SVWE $\sxx\in\FX(A)$, let
 $\yy\in \FX\esnu(A)$ be s.t. $\forall n$, $\|\yy_n-\bpsi_n(\xx)\|_{\N}\leq 2m I_n \mdset/\rho$ (cf. \Cref{lem:dist_agg_genized_sets}). Then
\begin{align}
 J & = \txt\sum_{n\in\N}  \big\langle  \cc(\hxxag)+\bh'_n, \ \hxx_n - \bpsi_n(\sxx)  \big\rangle \notag \\
 & =  \txt\sum_{n\in \N}  \big\langle  \cc(\hxxag)+\bh'_n , \ \hxx_n - \yy_n  \big\rangle \notag \\
 & \quad +  \txt\sum_{n\in\N}  \big\langle \cc(\hxxag)+\bh'_n , \ \yy_n - \bpsi_n(\sxx)    \big\rangle \notag \\
 &   \leq 0 + \txt \sum_{n\in\N} \Bdf\|\bpsi_n(\sxx)- \yy_n\|  \notag \\ 
&  \leq   \Bdf2M\tfrac{\mdset}{\rho} \, ,\label{eq:Jnu}
\end{align}

Let us summarize by combining \eqref{eq:th-firstbound} and \eqref{eq:Jnu}:
\begin{equation}\label{eq:unpperboundmono}
\begin{aligned}  
&\langle\cc(\sxxag)\hh-\hh\cc(\hxxag),\sxxag\hh-\hh\hxxag \rangle \hh + \hh \txt\sum_i \langle  \g'_i\hh- \hh\rr'_i ,  \sxx_i\hh- \hh \hxx_n \rangle \notag \\
&\leq 2M \left( 3 \tfrac{\Bdf}{\rho} \mdset +   \mduti \right) \ .
\end{aligned}
\end{equation}
Hence, if $H'$ is strongly monotone with modulus $\stgccvut$, then 
\begin{align*}
 & \stgccvut  \hh\norm{ \psi(\hxx) \-  \sxx }^2  \hh\leq \hh \txt\sum_i \langle  \g'_i\hh- \hh\rr'_i ,  \sxx_i\hh- \hh \hxx_n \rangle\hh \leq \hh2M \big( 3 \tfrac{\Bdf}{\rho} \mdset \hh+  \hh \mduti \big) \, .
 \end{align*}
 If $H'$ is aggregatively strongly monotone with modulus $\beta$, then 
 \begin{equation*}
  \beta \| \hxxag \hh- \hh\sxxag\!\|^2\hh\leq \hh \langle\cc(\sxxag\!)\hh-\hh\cc(\hxxag),\sxxag\hh-\hh\hxxag \rangle \hh\leq \hh 2M \big( 3 \tfrac{\Bdf}{\rho} \mdset \hh + \hh \mduti \big) \, .
  \end{equation*}


\begin{small}
\bibliographystyle{IEEEtran}
\bibliography{../../bib/shortJournalNames,../../bib/biblio1,../../bib/biblio2,../../bib/biblio3,../../bib/biblioBooks}
\end{small}
\end{document}